\documentclass[12pt]{article}

\usepackage{amsmath}
\usepackage{amssymb}
\usepackage{amsthm}
\usepackage{tikz}
\usetikzlibrary{decorations.markings,decorations.pathreplacing,calc}
\usetikzlibrary{arrows.meta}
\usepackage{enumerate}
\usepackage{comment}
\usepackage{cite}

\usepackage[margin=1in]{geometry}
\usepackage{subfig}
\usepackage[utf8]{inputenc}

\usepackage[colorlinks = true, citecolor=blue]{hyperref}
\usepackage{array}

\newcommand{\be}{\begin{equation}}
\newcommand{\ee}{\end{equation}}
\newcommand{\ba}{\begin{array}}
\newcommand{\ea}{\end{array}}
\newcommand{\bea}{\begin{eqnarray}}
\newcommand{\eea}{\end{eqnarray}}

\newcommand{\ra}{\rangle}
\newcommand{\la}{\langle}

\newcommand{\calL}{{\cal L }}

\newcommand{\calE}{{\cal E }}

\newcommand{\calO}{{\cal O }}

\newtheorem{dfn}{Definition}
\newtheorem{prop}{Proposition}

\newtheorem{lemma}{Lemma}

\newtheorem{theorem}{Theorem}

\newcommand{\footremember}[2]{%
    \footnote{#2}
    \newcounter{#1}
    \setcounter{#1}{\value{footnote}}%
}

\newcolumntype{C}[1]{%
 >{\vbox to 4ex\bgroup\vfill\centering}%
 p{#1}%
 <{\egroup}}  

\begin{document}

\title{Classical algorithms for quantum mean values}
\author{Sergey Bravyi  \footremember{ibmyorktown}{IBM Research, Yorktown Heights, NY 10598, 
U.S.A. }
	\and David Gosset  \footremember{iqc}{Institute for Quantum Computing, University of Waterloo, Canada.} 		    \footremember{co}{Department of Combinatorics and Optimization, University of Waterloo, Canada  }
\and Ramis Movassagh  \footremember{ibmcambridge}{IBM Research, Cambridge, MA 02142, U.S.A.}
}

\date{\today}

\maketitle

\begin{abstract}
We consider the task of estimating the expectation value of an $n$-qubit tensor product observable 
$O_1\otimes O_2\otimes  \cdots \otimes O_n$
in the output state of a shallow quantum circuit. 
This task is a cornerstone of variational quantum algorithms for optimization, machine learning, and the simulation of quantum many-body systems. Here we study its computational complexity for constant-depth quantum circuits and
three types of single-qubit observables $O_j$ which are
(a) close to the identity,
(b)  positive semidefinite,
(c)  arbitrary.
It is shown that the mean value problem admits a classical approximation algorithm
with runtime scaling as $\mathrm{poly}(n)$ and  $2^{\tilde{\calO}(\sqrt{n})}$
in cases (a,b) respectively. In case (c) we give a linear-time algorithm
for geometrically local circuits on a two-dimensional grid.
The mean value is approximated with a small relative error in case (a), while
in cases (b,c) we satisfy a less demanding additive error bound.
The algorithms are based on (respectively) Barvinok's polynomial interpolation method,
a polynomial approximation for the OR function arising from quantum query complexity,
and a Monte Carlo method combined with Matrix Product State techniques.
We also prove a technical lemma characterizing a zero-free region for certain polynomials associated with a quantum circuit, which may be of independent interest.
\end{abstract}
\section{Introduction}
In this work we consider the computation of expectation values at the output of a shallow quantum circuit. Suppose we are given an $n$-qubit quantum circuit $U$ of depth $d=\mathcal{O}(1)$ along with $n$ 
single-qubit operators $O_{1},\dots,O_{n}$. If each operator $O_j$ is Hermitian then the tensor product $O\equiv O_{1}\otimes O_{2}\otimes\cdots\otimes O_{n}$ is an observable and its mean value with respect to the state $U|0^n\rangle$ is given by
\be
\label{mean_value}
\mu\equiv\langle0^{n}|U^{\dagger}OU|0^{n}\rangle.
\ee
The mean value $\mu$ can be efficiently estimated on a quantum computer
by repeatedly preparing the state $U|0^n\ra$, measuring each single-qubit observable $O_j$, 
and averaging the product of the
measured eigenvalues.

The mean value problem, i.e., estimation of $\mu$, for tensor product observables
is a common step of NISQ\footnote{Noisy Intermediate Scale Quantum} era quantum algorithms~\cite{preskill2018quantum}
since the readout requires only single-qubit operations
which tend to be more reliable than two-qubit gates. For example, variational quantum algorithms such as VQE\footnote{Variational Quantum Eigensolver}~\cite{Peruzzo2014}
or QAOA\footnote{Quantum Approximate Optimization Algorithm}~\cite{farhi2014quantum}
aim at minimizing the expected energy $\la 0^n|U^\dag H U|0^n\ra$,
where $H$ is a Hamiltonian and $U$ is a shallow quantum circuit chosen from a suitable variational class.
In many interesting cases, such as quantum chemistry simulations~\cite{Peruzzo2014,Kandala2017,mcclean2016theory,nam2019ground},
the Hamiltonian $H$ can be
written as a linear combination of $\mathrm{poly}(n)$ Pauli operators, and so the expected energy $\la 0^n|U^\dag H U|0^n\ra$ is a sum
of $\mathrm{poly}(n)$ mean values $\mu$ of the form Eq.~\eqref{mean_value}. The mean value $\mu$ can also represent an  output probability
of the quantum circuit, i.e., the probability of observing a particular measurement outcome
if we prepare the state $U|0^n\ra$ and measure some of the qubits
in the standard basis. In this case each observable $O_j$ is either a projector
$|0\ra\la 0|$, $|1\ra\la 1|$ or the single-qubit identity operator. The estimation of output probabilities is a key step in variational quantum classifiers~\cite{schuld2019quantum,havlivcek2019supervised} where the mean value of the observable $O=|0\ra\la 0|^{\otimes n}$ encodes a single entry of the classifier kernel function. These and other quantum algorithms have prompted the development and demonstration of several error mitigation schemes which target a reliable mean value estimation in the presence of noise ~\cite{temme2017error,li2017efficient,li2018practical,otten2018recovering,bonet2018low,kandala2018extending}.

Given the current enthusiasm for variational quantum algorithms, it is natural to question whether or not they can be more powerful than classical algorithms in some sense. Unfortunately, empirical comparisons are limited by the fact that large scale quantum computers are not yet available. Furthermore, heuristic quantum algorithms are challenging to analyze mathematically and generally do not have performance guarantees. In this paper we investigate the computational power of variational quantum algorithms
based on constant-depth circuits 
 by asking whether or not they are ``easy" to simulate on a classical computer. \textit{Does the mean value problem require a quantum computer?} 

We note that the computational complexity of estimating mean values depends crucially on the type of approximation guarantee requested. One may ask for an
approximation $\tilde{\mu}$ which achieves a small \textit{additive error} $\delta$, i.e., $|\tilde{\mu} - \mu|\le \delta$. Here we assume that the single-qubit observables are normalized in the sense that $\|O_j\|\le 1$.  If $U$ is a polynomial-size circuit, this task is BQP-complete almost by definition. On the other hand, in the case of interest---where $U$ is described by a constant-depth quantum circuit---it captures the power of a restricted model of quantum computation, which consists of repeating a constant-depth quantum computation followed by single-qubit measurements a polynomial number of times and averaging the results. As discussed above, this describes a typical step of variational quantum algorithms.  

Alternatively, one may ask for a much more stringent approximation $\tilde{\mu}$ which achieves a small \textit{relative error} $\delta$, i.e., $e^{-\delta} \mu \leq |\tilde{\mu}-\mu|\leq e^{\delta} \mu$. This is clearly at least as difficult as computing an additive error approximation in the case of normalized observables. In fact, this task is \#P-hard, and therefore intractable, for a general constant-depth circuit $U$ and $\delta=\calO(1)$\footnote{ A standard reduction using postselected gate teleportation \cite{terhal2004adaptive} shows that the problem is just as hard as its generalization in which $U$ is given by any circuit of polynomial size , which is \#P-hard \cite{goldberg2017complexity}.}.   Thus, a classical algorithm which computes an additive error approximation of $\mu$ is simulating an efficient quantum computation, while a classical algorithm which computes a relative error approximation is solving a more challenging problem. 

In this paper we consider the complexity of the mean value problem as a function of circuit depth, qubit connectivity, the structure of observables $O_j$, and the type of approximation. We describe classical algorithms for three special cases as detailed below and summarized in Table \ref{table:comparison}. Our results clarify the circumstances in which variational quantum algorithms may provide a quantum advantage. Some good news is that the mean value problem with super-constant depth quantum circuits remains out of reach for classical computers as far as we know. However, constant-depth circuits on a 2D or 3D grid are not as powerful as may have previously been expected: the corresponding mean value problem can be solved classically in time $\calO(n)$ and $2^{\calO(n^{1/3})}$ respectively. We also find that, for general constant-depth circuits without geometric locality, mean value problems with positive semidefinite observables---including e.g., the additive error approximation of output probabilities---can be solved on a classical computer in subexponential time $2^{\calO(n^{1/2})}$. Our results suggest that  achieving a  quantum advantage 
with variational quantum algorithms requires either a super-constant circuit depth (e.g., $d\sim \log{n}$) or qubit connectivity graphs that cannot be  locally embedded in a 2D grid, or observables that cannot be expressed as linear combinations of $\mathrm{poly}(n)$
tensor product operators. 

\begin{table}[t]
\small
  \begin{center}
    \caption{Complexity of the mean value problem}
    \label{table:comparison}
    \begin{tabular}{|C{3.7cm}|c|c|c|}
\hline 
      \textbf{Quantum circuit} $U$ & \textbf{Observables} $O_j$ & \textbf{Relative error}  & \textbf{Additive error }\\
	      \hline
     Polynomial size & Pos. semidefinite & \#P-hard \cite{goldberg2017complexity} & BQP-complete\\
		\hline
     Constant depth & Close to $I$ &  P [Thm. \ref{thm:main}] & P [Thm. \ref{thm:main}]\\
		\hline
     Constant depth & Pos. semidefinite & \#P-hard \cite{goldberg2017complexity, terhal2004adaptive} &  \begin{tabular}{@{}c@{}}BQP \\ Subexp. classical  [Thm. \ref{thm:additive}]\end{tabular}\\
		\hline
     2D Constant depth & Hermitian &\begin{tabular}{@{}c@{}}\#P-hard \cite{goldberg2017complexity, terhal2004adaptive}  \\ Subexp. classical \cite{markov2008simulating}\end{tabular}&  BPP [Thm. \ref{thm:MPS}] \\
\hline
    \end{tabular}
  \end{center}
\end{table}
\subsection{Summary of results}
\paragraph{Algorithm (a): Each $O_j$ is close to the identity}
Our first result concerns the special case of the mean value problem where each 
observable $O_{j}$ is close to the identity operator in the sense that 
\begin{equation}
\|O_j-I\|\leq \calO(2^{-5d}).
\label{eq:condidentity}
\end{equation}
Recall that $d=\mathcal{O}(1)$ denotes the circuit depth. 
For observables satisfying Eq.~\eqref{eq:condidentity} we describe a classical deterministic algorithm
that approximates $\mu$ to within a relative error $\delta$. The runtime of the algorithm scales  
polynomially in the number of qubits $n$ and $\delta^{-1}$.
Note that while we are primarily interested in the case where $O_j$ are Hermitian, 
our algorithm is  not restricted to this case.

The condition Eq.~(\ref{eq:condidentity}) can be satisfied in the 
case of very noisy measurements. For example, suppose a bit-flip 
channel $\calE(\rho)=(1-p)\rho + p X\rho X$ 
is applied to each qubit immediately  before the measurement. 
Here $p\in [0,1/2]$ is the error rate.
Consider a noisy mean value
\[
\mu_p \equiv \la 0^n | \calE^{\otimes n} (U|0^n\ra\la 0^n|U^\dag) |0^n\ra.
\]
A simple calculation shows that $\mu_p=2^{-n} \mu$, where $\mu$ is 
the ideal mean value 
defined by Eq.~(\ref{mean_value}) with the observables
$O_j=I+(1-2p)Z$. 
Thus our algorithm  approximates the noisy mean value $\mu_p$
with a small relative error in the strong noise regime
\[
\frac12 - \calO(2^{-5d}) \le p \le \frac12.
\]
We envision that observables $O_j$ satisfying Eq.~(\ref{eq:condidentity})
could be measured for verification purposes while executing a variational
quantum algorithm. Indeed, a typical step of such an algorithm
repeatedly prepares a variational
state $U|0^n\ra$ and measures all qubits in the standard basis. 
The measurement data collected by the quantum algorithm 
can be used to approximate the mean
value $\mu$ defined in Eq.~(\ref{mean_value}) for any observables $O_j$ diagonal in the standard basis,
for example,  $O_j=e^{i\theta Z}$.
The verification step would compare the mean value $\mu$ inferred 
from the measurement data and the approximation $\tilde{\mu}$
computed by the classical algorithm. 
The latter can be computed efficiently whenever $O_j$ obeys Eq.~(\ref{eq:condidentity}),
that is, $|\theta|=\calO(2^{-5d})$.
An attractive feature of this method is that the 
verification step and the algorithm that is being verified
access the same measurement data. Thus no additional
quantum operations are required. 

The algorithm works by classically computing mean values 
$\mu_S=\la 0^n|U^\dag \prod_{j\in S} O_j U|0^n\ra$ for all 
subsets of qubits $S$ of size up to $\calO(\log{(\delta^{-1}n)})$
satisfying a suitable connectivity property. We show that the number of such 
subsets is at most $\mathrm{poly}(n)$.
Each mean value $\mu_S$ can be computed by restricting the circuit $U$
onto the ``lightcone" of $S$. It is shown that the restricted circuit
can be simulated classically in time $\mathrm{poly}(n)$. 
The desired approximation $\tilde{\mu}$ is obtained by combining 
the mean values $\mu_S$ using 
the polynomial interpolation lemma due to Barvinok \cite{barvinok2016combinatorics}.
To this end, we define a degree-$n$ polynomial $f(\epsilon)=\langle0^{n}|U^{\dagger}O(\epsilon)U|0^{n}\rangle$,
where  $O(\epsilon)$ is the tensor product of observables
$O_{j}(\epsilon)=I+\epsilon(O_{j}-I)$.
Note that $f(0)=1$ and $f(1)=\mu$.

Barvinok's lemma implies that if the polynomial $f(\epsilon)$ is zero-free
in a disk $|\epsilon|\le \beta$ for some constant
$\beta>1$, then $\log f(1)=\log\mu$ can be approximated with an additive error $\delta$ from the Taylor expansion of $\log f(\epsilon)$
at $\epsilon=0$ truncated at the order $p=\calO(\log{(\delta^{-1}n)})$.
We show that the coefficients of the Taylor series of $\log f(\epsilon)$
to the $p$-th order are simply related to the mean values $\mu_S$
computed at the first stage of the algorithm. 
The main technical step in applying Barvinok's lemma is establishing the
zero-freeness condition.  For a depth-$d$ quantum circuit composed of two-qubit gates
we show that $f(\epsilon)$ is zero-free in a disk of radius  $\beta=\epsilon_0/\gamma$, where
$\gamma\equiv \max_j \|O_j-I\|$ and 
 $\epsilon_0=\Omega(2^{-5d})$.
 We prove this by  constructing a probability distribution over $2n$-bit strings $p_{\epsilon}(z)$ 
 such that $p_{\epsilon}(0^{2n})$ is proportional to $|f(\epsilon)|^2$.  The Lov\'{a}sz Local Lemma is 
 then applied to show that the probability $p_{\epsilon}(0^{2n})$ is strictly positive
 for $|\epsilon|\leq \beta$. This proves $f(\epsilon)\ne 0$ for $|\epsilon|\le \beta$.
The inverse exponential scaling of $\epsilon_0$ with $d$ is shown to be  optimal.
On the other hand, we show that a random unitary $U$ satisfying the $2$-design property 
typically has a much larger zero-free radius
$\beta=\epsilon_0/\gamma$, where 
$\epsilon_0\ge 1-\calO\left(\frac{\log(n)}{n}\right)$. This result suggests that the worst-case bound on
the applicability region of our algorithm established in Eq.~(\ref{eq:condidentity})
is unlikely to be tight for the vast majority of circuits. 

We note that a similar algorithm, also based on Barvinok's approach, was previously proposed for approximating output probabilities of IQP circuits composed of gates which are sufficiently close to the identity \cite{mann2018approximation}.

\paragraph{ Algorithm (b): Each $O_j$ is positive semidefinite}

Our next result is a classical algorithm that approximates the mean value 
$\mu$ to within a given additive error $\delta$, for general constant-depth circuits $U$ and
positive semidefinite tensor product observables. More precisely, we assume
that $\|O_j\|\le 1$ and  $O_j\ge 0$ for all $j$. The algorithm has runtime
exponential in $\sqrt{n\log(\delta^{-1})}$, with a prefactor logarithmic in $n$.

Our result also sheds some light on the more general additive error mean value problem for shallow quantum circuits. For more general Hermitian observables $O_j$ which may not be positive semidefinite, our algorithm outputs an approximation to the \textit{absolute value} $|\langle 0^n|U^{\dagger} O U|0^n\rangle|$ to within an additive error $\delta$. Thus our algorithm would provide an additive error estimate to the mean value \textit{if we only knew the sign!}  It is an open question whether or not this more general case admits a subexponential classical algorithm.

The algorithm is based on approximating the projector onto the output state $|\psi\rangle\equiv U|0^n\rangle$ of the quantum circuit by an operator which has a subexponential classical description size. Similar ideas were used in Ref.~\cite{eldar2017local} to establish a certain expansion property of the probability distribution obtained by measuring $\psi$ in the standard basis. To explain the main idea, let us specialize to the case 
$O=|0^n\rangle\langle0^n|$ in which we aim to estimate an output probability
\begin{equation}
\mu=|\langle 0^n|\psi\rangle |^2
\label{eq:prob}
\end{equation}
of a shallow quantum circuit. The output state $\psi$ of the quantum circuit is the unique state which is orthogonal to each of the commuting projectors $U|1\rangle\langle 1|_j U^{\dagger}$ for $1\leq j\leq n$. These operators are simultaneously diagonalized in the basis $\{|\hat{z}\rangle=U|z\rangle: z\in \{0,1\}^n\}$. The projector $I-|\psi\rangle \langle \psi|$ has the property that it computes the multivariate OR function in this basis, in the sense that
\begin{equation}
\left(I-|\psi\rangle \langle \psi|\right) |\hat{z}\rangle=\mathrm{OR}(z)|\hat{z}\rangle,
\label{eq:orfunction}
\end{equation}
where $\mathrm{OR}(z)$ is zero iff $z=0^n$. Consequently, we obtain an $\delta$-approximation $P$ to the projector $|\psi\rangle \langle \psi|$ (in the spectral norm) by plugging in an $\delta$-approximation of the multivariate OR function on the RHS of Eq.~\eqref{eq:orfunction}. Using an optimal polynomial approximation derived from quantum query complexity \cite{buhrman1999bounds, de2008note}, one obtains an operator $P$ which is a sum of $2^{\tilde{\calO}(\sqrt{n\log(\delta^{-1})})}$ terms each acting nontrivially on $2^{\tilde{\calO}(\sqrt{n\log(\delta^{-1})})}$ qubits. The algorithm outputs the estimate $\tilde{\mu}=\langle 0^n|P|0^n\rangle$ of the mean value Eq.~\eqref{eq:prob}, which can be computed exactly in time $2^{\tilde{\calO}(\sqrt{n\log(\delta^{-1})})}$.

\paragraph{Algorithm (c): The circuit is geometrically local in 2D or 3D}

Our final result is a classical randomized  algorithm
that approximates the mean value $\mu$ defined in Eq.~(\ref{mean_value}) to within an additive error $\delta$ 
for any single-qubit operators $O_j$ satisfying $\|O_j\|\le 1$. This algorithm only applies to constant-depth geometrically local quantum circuits,
i.e., circuits with nearest-neighbor gates on a $D$-dimensional grid of qubits. From a practical perspective, the most interesting cases are $D=2$ and $D=3$.
In the 2D case  our algorithm achieves a polynomial  runtime  $\calO(n\delta^{-2})$.
The scaling with $n$ is optimal since one needs a time linear in $n$ simply to examine
each gate in the circuit.  However, the $\calO$ notation hides a very large constant
factor that limits practical applications of the algorithm.
For comparison, state-of-the-art tensor network 
simulators~\cite{markov2008simulating,pednault2017breaking,boixo2017simulation,villalonga2019establishing}
enable simulation of medium size constant-depth 2D circuits 
with $n\sim 100$
but have a super-polynomial asymptotic runtime $2^{\calO(\sqrt{n})}$.
In the 3D case our algorithm achieves a sub-exponential runtime $\delta^{-2} 2^{\calO(n^{1/3})}$.
We believe that accomplishing the same simulation using  the standard tensor network methods~\cite{markov2008simulating,aaronson2017complexity}
would require time $2^{\calO(n^{2/3})}$.

The main idea behind the algorithm is to express the mean value
as  $\mu =\la \Psi_0|W|\Psi_1\ra$, where
$\Psi_0,\Psi_1$ are Matrix Product States  (MPS) of $n$ qubits with bond dimension
$\calO(1)$, and $W$ is a permutation of $n$ qubits.
We then approximate $\mu$ using a Monte Carlo algorithm similar to the one 
proposed by Van den Nest~\cite{Nest2009}. 
It is based on the identity
\[
\mu = \la \Psi_0|W|\Psi_1\ra = \sum_{x\in \{0,1\}^n} \; \pi(x) F(x),
\]
where $\pi(x)=|\la x|\Psi_0\ra|^2$ and $F(x)=\la x|W|\Psi_1\ra \la x|\Psi_0\ra^{-1}$.
Using the standard MPS algorithms one can compute the quantity $F(x)$
for any given $x$ in time $\calO(n)$. Likewise, one can sample $x$ from the probability
distribution $\pi(x)$ in time $\calO(n)$. A simple calculation shows that the variance
of a random variable
$F(x)$ with $x$ drawn from $\pi(x)$ is  at most one. 
Thus one can approximate $\mu$ by an empirical mean value
$\tilde{\mu}=S^{-1}\sum_{i=1}^S F(x^i)$, where $x^1,\ldots,x^S$ are 
independent samples from $\pi(x)$ and $S=\calO(\delta^{-2})$.
The 3D simulation algorithm follows the same idea except that the
required MPS bond dimension is $2^{\calO(n^{1/3})}$.

\subsection{Open problems}

A central open question raised by this work is whether the quantum mean value problem 
can be solved efficiently on a classical computer
in the case of shallow circuits, tensor product observables, and additive approximation error.
Alternatively, can we provide some evidence that this task is classically hard? This question directly addresses the computational power of variational quantum algorithms based on shallow circuits.

 To shed light on this problem, one may ask whether our algorithms can be improved or generalized. For example, can the subexponential algorithm for positive semi-definite observables be generalized to Hermitian tensor-product observables? Can the additive error mean value problem for 3D shallow circuits be solved in polynomial time on a classical computer? Can the runtime of our algorithms for the 2D and 3D shallow circuits be reproduced using 
 simulators based on tensor network contraction~\cite{markov2008simulating}?
 
 Another interesting question is whether large-scale instances of the quantum mean value problem
 can be solved by hybrid quantum-classical algorithms with limited quantum resources
 (e.g. small number of qubits).  For example,
a promising class of hybrid algorithms known as holographic quantum simulators was 
recently proposed~\cite{kim2017holographic,kim2017noise}.
Loosely speaking,  such algorithms enable a simulation of 2D lattice models on a 1D quantum computer by converting one spatial dimension into time. We anticipate that a similar approach can be
used to solve $n$-qubit instances of the quantum mean value problem with 2D shallow circuits
on a quantum computer with only $O(n^{1/2})$ qubits. 
Even though Theorem~\ref{thm:MPS} provides a purely classical linear time algorithm
for the problem, its runtime has a very unfavorable scaling with the circuit depth.
Hybrid algorithms may potentially remedy this inefficiency.

One may also further probe the complexity of the relative error mean value problem for shallow circuits. While this problem is known to be \#P-hard in the worst case, it is interesting to elucidate broad classes of quantum circuits for which the problem can be solved efficiently. 
For example, it can be easily shown that the mean value problem admits a  polynomial time
classical algorithm for Pauli-type observables and quantum circuits that belong to the 3rd level of
the Clifford hierarchy~\cite{gottesman1999demonstrating}.
At the same time, approximating output probabilities of such circuits is known to be
 \#P-hard in the worst case~\cite{bremner2016average}.

Random quantum circuits are another possible avenue to explore. 
For example, it has been conjectured that relative error approximation of output probabilities is \#P-hard 
for random quantum circuits of sufficiently high depth~\cite{boixo2018characterizing,bouland2019complexity,movassagh2019cayley}. Does our bound Eq.~\eqref{eq:randbound} on the zero-free disk for the polynomial $f(\epsilon)$ when $U$ is a random circuit have any bearings on this conjecture? Finally, it may be possible to improve our lower bound $\epsilon_0=\Omega(2^{-5d}\gamma^{-1})$ on the zero-free radius for depth-$d$ circuits, although it cannot be improved beyond $\Omega(2^{-d})$ due to the example described in Section \ref{sec:zerofree}.

\section{Notation}

Let $[n]=\{1,2,\dots,n\}$. Given a subset of qubits $S\subseteq[n]$,
let $\mathcal{A}_{S}$ be an operator algebra that consists of all $n-$qubit operators that only act nontrivially on $S$; also let
$\mathcal{A}_{j}\equiv\mathcal{A}_{\{j\}}$. Consider a fixed unitary $U$. 
\begin{dfn}
The (forward) lightcone of a qubit $j$, denoted by $\mathcal{L}(j)$, is the
smallest subset of qubits $\mathcal{L}(j)\subseteq[n]$ such that
$j\in\mathcal{L}(j)$ and $U^{\dagger}\mathcal{A}_{j}U\subseteq\mathcal{A}_{\mathcal{L}(j)}$. For any subset $S\subseteq [n]$ we define 
\[
\mathcal{L}(S)=\cup_{j\in S} \mathcal{L}(j).
\]
We also define the backward lightcone $\mathcal{L}_{\leftarrow}(S)$ by replacing $U$ with $U^{\dagger}$ in the above. 
\end{dfn}
Therefore, for any $S\subset [n]$, the unitary $U$ maps any operator acting
on $S$ to an operator supported on $\mathcal{L}(S)$, and the circuit $U^{\dagger}$ maps any operator supported on $S$ to an operator supported on $\mathcal{L}_\leftarrow (S)$.

The forward and backward lightcones have the following symmetry.
\begin{prop}
Let $j,k\in [n]$. Then $j\in \mathcal{L}(k)$ if and only if $k\in \mathcal{L}_{\leftarrow}(j)$.
\label{prop:sym}
\end{prop}
\begin{proof}
The statement is clearly true if $j=k$, so consider the case $j\neq k$. Below we show that $k\notin \mathcal{L}_{\leftarrow}(j)$ implies $j\notin \mathcal{L}(k)$, which establishes the ``only if" direction. The ``if" direction then follows as it is the same statement with $U$ replaced by $U^{\dagger}$. 

So suppose $k\notin \mathcal{L}_{\leftarrow}(j)$. Equivalently, any operator in $\mathcal{A}_k$ commutes with any operator in $U\mathcal{A}_{j}U^{\dagger}$.  Equivalently, any operator in $U^{\dagger}\mathcal{A}_k U$ commutes with any operator in $\mathcal{A}_{j}$, which is the statement that $j\notin \mathcal{L}(k)$.
\end{proof}
It will also be convenient to define iterated lightcones.
\begin{dfn}
Given a unitary $U$, define iterated forward and backward lightcones  of $S\subseteq [n]$
\begin{align*}
\mathcal{L}(S,1)&=\mathcal{L}(S) \qquad  &&\mathcal{L}_{\leftarrow}(S,1)=\mathcal{L}_{\leftarrow}(S)\\
\mathcal{L}(S,2)&=\mathcal{L}_{\leftarrow}(\mathcal{L}(S))\qquad  &&\mathcal{L}_\leftarrow(S,2)=\mathcal{L}(\mathcal{L}_{\leftarrow}(S))\\
\mathcal{L}(S,3)&=\mathcal{L}(\mathcal{L}_{\leftarrow}(\mathcal{L}(S))) \qquad &&\mathcal{L}_{\leftarrow}(S,3)=\mathcal{L}_{\leftarrow}(\mathcal{L}(\mathcal{L}_{\leftarrow}(S)))\\
& \vdots && \qquad \vdots
\end{align*}
\end{dfn}
We also define the maximum iterated lightcone sizes 
\begin{dfn}
For each positive integer $c$ define
\[
\ell_{c}=\max\{\max_{1\leq j\leq n}|\mathcal{L}(j,c)|, \max_{1\leq j\leq n}|\mathcal{L}_{\leftarrow}(j,c)| \}
\]
\end{dfn}
The quantities $\ell_{c}$ quantify the growth of the lightcone under repeated applications of $U$ or $U^{\dagger}$.  Clearly we have the upper bound
\begin{equation}
\ell_{c}\leq (\ell_{1})^{c}.
\label{eq:naive}
\end{equation}
Indeed, if $U$ is a depth-$d$ circuit composed of two-qubit gates then 
\begin{equation}
\ell_{c}\leq 2^{cd} \qquad c\geq 1, \qquad \textbf{ (lightcone growth, depth-$d$ circuit)}
\label{eq:ellconst}
\end{equation}
In some cases Eq.~\eqref{eq:naive} is a poor upper bound. For example, if all gates are restricted to be nearest-neighbor two-qubit gates on a $D$-dimensional grid then we have
\begin{equation}
\ell_{c}\leq (2cd)^{D} \qquad c\geq 1,  \qquad \textbf{ (lightcone growth, $D$-dimensions)}
\label{eq:elld}
\end{equation}
For our purposes the most important distinguishing feature of constant depth circuits is that $\ell_{c}=\mathcal{O}(1)$ for any constant $c$.

\section{Simulation by polynomial interpolation}

Let us define the polynomial
\begin{equation}
f(\epsilon)=\langle0^{n}|U^{\dagger}O(\epsilon)U|0^{n}\rangle\label{eq:f_epsilon}
\end{equation}
where $\epsilon\in\mathbb{C}$ and $O(\epsilon)\equiv\bigotimes_{j=1}^{n}O_{j}(\epsilon)$
and
\[
O_{j}(\epsilon)=I+\epsilon(O_{j}-I).
\]
Clearly, $f(0)=1$ and $f(1)=\mu$ is the quantity we wish to approximate. Our main result is the following theorem.
\begin{theorem}
There exists a deterministic classical algorithm that takes as inputs
a quantum circuit $U$ acting on $n$ qubits, an error
tolerance $\delta>0$, and a product operator $O=\bigotimes_{j=1}^{n}O_{j}$
such that 
\begin{equation}
\|O_{j}-I\|\le \frac{1}{60\beta\cdot \ell_{1}\cdot \ell_{4}}
\label{eq:ocond}
\end{equation}
for all $j$, where $\beta>1$ is an absolute constant. The algorithm outputs a complex number $\tilde{\mu}$
that approximates $\mu=\langle0^{n}|U^{\dagger}OU|0^{n}\rangle$ with
a multiplicative error $\delta$, that is
\be
\label{log_approx}
|\log\mu-\log\tilde{\mu}|\le\delta.
\ee
The running time of the algorithm is $(n\delta^{-1})^{\mathcal{O}(\ell_{1})}$.
\label{thm:main}
\end{theorem}

For constant-depth circuits $d=\mathcal{O}(1)$ we have $\ell_{1}, \ell_{4}=\mathcal{O}(1)$ and we obtain the claimed efficient algorithm to compute $\mu$. For a general depth-$d$ circuit composed of two-qubit gates $U$ or a geometrically local circuit in $D$-dimensions we may plug in  Eq.~\eqref{eq:ellconst} or Eq.~\eqref{eq:elld} respectively to see how the runtime and the condition Eq.~\eqref{eq:ocond} depend on depth $d$.

To prove Theorem \ref{thm:main}, we use a zero-free region lemma and Barvinok's interpolation lemma.
\begin{lemma}
\label{lem:(zero-free-region)}(Zero-free region) Let $U$ be a quantum circuit, $O=\bigotimes_{j=1}^{n}O_{j}$ be a product operator, and let $\gamma=\max_j \| O_j - I\|$. The polynomial $f(\epsilon)$ is zero-free on the disk $|\epsilon|\le \epsilon_0$, where 
\begin{equation}
\epsilon_0 = \frac{1}{60\gamma\cdot \ell_{1}\cdot \ell_{4}}.
\label{eq:epsdisc}
\end{equation}
\end{lemma}
The proof of Lemma \ref{lem:(zero-free-region)} is deferred to Section \ref{sec:zerofree}. By choosing $\gamma$ small enough as stated in the main theorem, we are guaranteed that $f(\epsilon)\ne 0$ on a disk of radius $|\epsilon|\le\epsilon_0=\beta$. Using Barvinok's lemma (lemma~\ref{lem:(Barvinok-lemma)}) we can interpolate between $f(0)$ and $f(1)$.

Below we shall write $g^{(m)}\equiv g^{(m)}(0)$ for the $m$-th derivative of a function
$g(\epsilon)$ evaluated at $\epsilon=0$.
Let us agree that $g^{(0)}=g(0)$.
\begin{lemma}
\label{lem:(Barvinok-lemma)}(Barvinok's interpolation lemma \cite{barvinok2016combinatorics,barvinok2017computing})
Let $f(\epsilon)$ be a polynomial of degree $n$ and suppose $f(\epsilon)\ne0$
for all $|\epsilon|<\beta$ , where $\beta>1$ is a real number. Let
us choose a branch of
\[
g(\epsilon)=\ln f(\epsilon)\qquad\text{for }|\epsilon|\le1
\]
and consider the its Taylor polynomial
\[
T_{p}(\epsilon)=g^{(0)}+\sum_{k=1}^p \frac{\epsilon^k}{k!} g^{(k)}.
\]
Then
\[
|g(\epsilon)-T_{p}(\epsilon)|\le\frac{n\beta^{-p}}{(p+1)(\beta-1)}\quad\text{ for all  }\quad|\ensuremath{\epsilon}|\le1.
\]
\end{lemma}
Assuming that $\beta>1$ is fixed a priori (below we use $\beta=2$) and setting 
\[
\tilde{\mu}=\exp[T_{p}(1)]
\]
one can achieve the bound Eq.~(\ref{log_approx}) by choosing
\[
p=\mathcal{O}(\ln n\delta^{-1}),
\]
where $\mathcal{O}$ depends only on $\beta$.

To complete the proof of Theorem \ref{thm:main}, in the remainder of this section we show that $\tilde{\mu}$ can be computed using runtime $(n\delta^{-1})^{\mathcal{O}(\ell_{1})}$. As was shown by Barvinok~\cite{barvinok2016combinatorics}, the derivatives of $g(\epsilon)$ can be obtained
from those of $f(\epsilon)$ by solving a simple linear system.
Indeed,  start with the identity 
$f'(\epsilon) = f(\epsilon) g'(\epsilon)$.
Taking the derivatives using  the Leibniz rule 
and setting $\epsilon=0$ one gets
\be
\label{der2}
f^{(m)}=\sum_{j=0}^{m-1}
{m-1 \choose j}
f^{(j)}g^{(m-j)}, \qquad m=1,\ldots,p.
\ee
This is a triangular linear system that determines
$g^{(1)},\ldots,g^{(p)}$ in terms of $f^{(1)},\ldots,f^{(p)}$
\begin{eqnarray}
g^{(1)} & = & f^{(1)}\\
g^{(2)} & = & f^{(2)} - f^{(1)}g^{(1)}\\
g^{(3)} & = &f^{(3)} - f^{(2)}g^{(1)} - 2f^{(1)}g^{(2)}
\end{eqnarray}
and so on. Here we noted that $f^{(0)}=f(0)=1$.

It remains to  calculate the derivatives $f^{(1)},\ldots,f^{(p)}$. To do so, it is convenient
to first define $\boldsymbol{\epsilon}=(\epsilon_{1},\dots,\epsilon_{n})\in\mathbb{C}^{n}$
and consider the multivariate version of Eq.~\eqref{eq:f_epsilon} and
then evaluate the results at $\boldsymbol{\epsilon}=(\epsilon,\epsilon,\dots,\epsilon)$.
To this end let
\begin{equation}
f(\boldsymbol{\epsilon})=\langle0^{n}|U^{\dagger}O(\boldsymbol{\epsilon})U|0^{n}\rangle,\label{eq:MultiVar_f}
\end{equation}
where $O(\boldsymbol{\epsilon})=\bigotimes_{j=1}^{n}O_{j}(\epsilon_{j})$
and
\[
O_{j}(\epsilon_{j})=I+\epsilon_{j}(O_{j}-I).
\]
 A monomial is defined by $M(\boldsymbol{\epsilon})=\alpha\prod_{j=1}^{n}\epsilon_{j}^{m_{j}}$,
where $\alpha\ne0$ is a complex coefficient and all $m_{j}\ge0$
are integers. We say $M(\boldsymbol{\epsilon})$ is supported on the
set $S\subseteq[n]$ if and only if $m_{j}>0$ for all $j\in S$ and
$m_{j}=0$ for all $j\notin S$. The degree of $M(\boldsymbol{\epsilon})$
is defined by $\sum_{j=1}^{n}m_{j}$. Let $T_p(\boldsymbol{\epsilon})$ be the Taylor series for $g(\boldsymbol{\epsilon})=\ln{f(\boldsymbol{\epsilon})}$ at $\boldsymbol{\epsilon}=0^n$ truncated at the order $p=\log_2{(n\delta^{-1})}$. By definition, the series $T_p(\boldsymbol{\epsilon})$ is a sum of monomials with degree at most $p$.
\begin{dfn}
Consider a subset $S\subseteq [n]$.
Define $g_S(\epsilon)$ as the sum of all monomials in $T_p(\boldsymbol{\epsilon})$ that are supported on
$S$, evaluated at the point $\boldsymbol{\epsilon}=(\epsilon,\ldots,\epsilon)$.
Define $h_S(\epsilon)$ as the sum of all monomials in $T_p(\boldsymbol{\epsilon})$
that are supported on some subset of $S$,
 evaluated at the point $\boldsymbol{\epsilon}=(\epsilon,\ldots,\epsilon)$.
\end{dfn}
Below we use the convention $g_\emptyset(\epsilon)=0$ and $h_\emptyset(\epsilon)=0$.
By definition, $g_S(\epsilon)$ and $h_S(\epsilon)$ are polynomials of degree at most $p$.
Let us discuss some basic properties of $g_S(\epsilon)$ and $h_S(\epsilon)$.
First, since the support of a monomial is uniquely defined and
$g_S(\epsilon)$ gets contributions from monomials of degree at least $|S|$, we have
\be
\label{p1}
T_p(\epsilon)=\sum_{S\subseteq [n]\, : \, 1\le |S|\le p} \; g_S(\epsilon).
\ee
Thus the task of computing $T_p(\epsilon)$ reduces to computing
$g_S(\epsilon)$ for all subsets $S$ of size at most $p$. First, 
we claim $g_S(\epsilon)$ can be computed in terms of $h_S(\epsilon)$
as follows.
\begin{prop}[\bf Inclusion-Exclusion]
\label{prop:IE}
\be
\label{IE}
g_S(\epsilon)=\sum_{T\subseteq S} (-1)^{|S\setminus T|} \, h_T(\epsilon)
\ee
for any subset $S\subseteq [n]$.
\end{prop}
\begin{proof}
Indeed, let us prove Eq.~\eqref{IE} by induction in $|S|$.
The base case is $|S|=1$. Then $g_S(\epsilon)=h_S(\epsilon)$ by definition.
Suppose we have already proved Eq.~\eqref{IE} for all subsets $S$ of size
$|S|\le m$. Let $S$ be a subset of size $m+1$. Then by definition, 
\be
g_S(\epsilon) = h_S(\epsilon) - \sum_{T\subset S} g_T(\epsilon),
\ee
where the sum runs over all proper subsets of $S$. Since $|T|\le m$,
we use the induction hypothesis to express $g_T(\epsilon)$ in terms of $h_R(\epsilon)$ with $R\subseteq T$. It gives
\be
g_S(\epsilon) = h_S(\epsilon) - \sum_{T\subset S} \sum_{R\subseteq T} (-1)^{|T\setminus R|} h_R(\epsilon).
\ee
Changing the summation order one gets
\be
g_S(\epsilon)  = h_S(\epsilon) - \sum_{R\subseteq S} h_R(\epsilon) \sum_{R\subseteq T\subset S} (-1)^{|T\setminus R|}.
\ee
Let us add and subtract the term $(-1)^{|T\setminus R|}$ with $T=S$. We get
\be
g_S(\epsilon)  = h_S(\epsilon) -  \sum_{R\subseteq S} h_R(\epsilon)
\left[ -(-1)^{|S\setminus R|} +  \sum_{R\subseteq T\subseteq S} (-1)^{|T\setminus R|} \right].
\ee
Using the well-known identity
\be
\sum_{T\, : \, R\subseteq T\subseteq S} \; (-1)^{|T\setminus R|} = \left\{
\ba{rcl}
1 &\mbox{if} & R=S,\\
0 && \mbox{otherwise} \\
\ea 
\right.
\ee
one arrives at
\be
g_S(\epsilon) = h_S(\epsilon) -h_S(\epsilon) + \sum_{R\subseteq S} h_R(\epsilon) (-1)^{|S\setminus R|}
=  \sum_{R\subseteq S} h_R(\epsilon) (-1)^{|S\setminus R|}.
\ee
This proves the induction hypothesis.
\end{proof}
Next  we claim that 
$h_S(\epsilon)$ can be computed for any given subset $S$ in time $2^{\mathcal{O}(\ell_{1}|S|)}$.
Indeed, define a polynomial
\be
\mu_S(\epsilon)\equiv  \la 0^n|U^\dag \prod_{j\in S} O_j(\epsilon) U |0^n\ra.
\ee
Let $T_{p,S}(\epsilon)$ be the Taylor expansion of $\ln{\mu_S(\epsilon)}$
at $\epsilon=0$
truncated at the $p$-th order. Note that 
\be
h_S(\epsilon)=T_{p,S}(\epsilon).
\ee 
Indeed, both polynomials are obtained from $g(\boldsymbol{\epsilon})$ by retaining
monomials of degree at most $p$ supported on some subset of $S$
and then setting $\boldsymbol{\epsilon}=(\epsilon,\ldots,\epsilon)$.
We claim that
the polynomial $\mu_S(\epsilon)$ can be computed in time roughly $2^{\mathcal{O}(\ell_{1}|S|)}$.
Indeed, one can first restrict the circuit $U$ by removing any gate which acts outside of $\calL(S)$. The latter contains at most $\ell_{1}|S|$ qubits.
The restricted circuit can be simulated by the brute-force method in time $2^{\mathcal{O}(\ell_{1}|S|)}$. Once the polynomial $\mu_S(\epsilon)$ is computed, 
one can solve the triangular linear system expressing $T_{p,S}(\epsilon)$
in terms of the first $p$ coefficients of $\mu_S(\epsilon)$
using Barvinok's method~\cite{barvinok2016combinatorics}.

Finally, we claim that $g_S(\epsilon)=0$ unless $S$ has a certain connectivity property.
\begin{dfn}
A subset $S\subseteq [n]$ is said to be $\calL$-connected  if any partition $S=S_1S_2$ 
into disjoint non-empty subsets $S_1,S_2$ 
satisfies $\calL(S_1)\cap \calL(S_2)\ne \emptyset$.
\end{dfn}
\setcounter{lemma}{2}
\begin{lemma}
$g_S(\epsilon)=0$ unless $S$ is $\calL$-connected.
\end{lemma}
\begin{lemma}
The number of $\calL$-connected subsets $S\subseteq [n]$ of size $p$
 is at most $n(3\ell_{2})^{p-1}$.
\end{lemma}
Combining Lemmas~3  and Eq.~\eqref{p1}, we infer that computing $T_p(\epsilon)$
amounts to computing  $g_S(\epsilon)$ for each
$\calL$-connected subset $S$ of size at most $p$.
By Lemma~4, 
the number of such subsets is at most
\be
n\sum_{q=1}^p (3\ell_{2})^{q-1} = n \frac{(3\ell_{2})^p-1}{3\ell_{2}-1} \leq \frac{n(3\ell_{2})^{p-1}}{1-1/(3\ell_{2})}\le \frac{3n}{2}(3\ell_{2})^{p-1}.
\ee
From Proposition~\ref{prop:IE} one infers that computing 
$g_S(\epsilon)$ for a given subset $S$ of size $|S|\le p$
amounts to computing  $h_T(\epsilon)$ for all $T\subseteq S$.
The number of subsets $T\subseteq S$ is 
$2^{|S|}\le 2^p$. As shown above, one can compute 
$h_T(\epsilon)$ for any given subset $T$ in time
roughly $2^{\mathcal{O}(\ell_{1}|T|)}\le 2^{\mathcal{O}(\ell_{1}p)}$. 
Thus the overall runtime required to compute $T_p(\epsilon)$ scales as
\[
\frac{3n}{2}(3\ell_{2})^{p-1} \cdot 2^p \cdot 2^{\mathcal{O}(\ell_{1}p)}\leq n(3\ell_{1}^2)^{p-1} 2^{\mathcal{O}(\ell_{1}p)}=(n\delta^{-1})^{\mathcal{O}(\ell_{1})},
\]
where in the first inequality we used Eq.~\eqref{eq:naive}.

\begin{proof}[\bf Proof of Lemma~3]
Suppose $S$ is not $\calL$-connected. Choose a partition 
$S=S_1 S_2$ such that $S_1,S_2$ are disjoint non-empty
subsets and $\calL(S_1)\cap \calL(S_2)=\emptyset$. Define
a multi-variate polynomial
\be
\mu_S(\boldsymbol{\epsilon}) = \la 0^n|U^\dag \prod_{j\in S} O_j(\epsilon_j) U |0^n\ra.
\ee
Since the lightcones of $S_1$ and $S_2$ do not overlap,
$\mu_S(\boldsymbol{\epsilon})$ is a product of some polynomial
depending on $\{\epsilon_j \, : \, j\in S_1\}$
and some polynomial depending on $\{\epsilon_j \, : \, j\in S_2\}$.
By definition, $g_S(\epsilon)$ is obtained from
the Taylor series of $\ln{\mu_S(\boldsymbol{\epsilon})}$ at $\boldsymbol{\epsilon}=0^n$
by retaining all monomials of degree $1,2,\ldots,p$ supported on $S$
and setting $\boldsymbol{\epsilon}=(\epsilon,\ldots,\epsilon)$.
However, since $\ln{\mu_S(\boldsymbol{\epsilon})}$ is a sum of some
function depending on $\{\epsilon_j \, : \, j\in S_1\}$ and
some function depending on $\{\epsilon_j \, : \, j\in S_2\}$,
the Taylor series of $\ln{\mu_S(\boldsymbol{\epsilon})}$ contains no monomials
supported on $S$. Thus $g_S(\epsilon)=0$, as claimed.
\end{proof}

\begin{proof}[\bf Proof of Lemma~4]
Define a graph $G$ with the set of vertices $[n]$ such that vertices
$i,j$ are connected by an edge iff 
\[
\mathcal{L}(i)\cap \mathcal{L}(j)\neq \emptyset.
\]
Using Proposition \ref{prop:sym} we see that this condition implies $j\in \mathcal{L}_{\leftarrow}(\mathcal{L}(i))=\mathcal{L}(i,2)$. Therefore the graph $G$ has maximum vertex degree at most $\ell_{2}$. By definition, a subset $S$ is $\calL$-connected iff
$S$ is a connected subset of vertices in $G$.
The number of connected subsets $S\subseteq [n]$ of size $p$ that contain a given
vertex $j$ is at most $(e\ell_{2})^{p-1}$, where $e=\exp{(1)}$, see Lemma~5 in 
\cite{aliferis2007accuracy}. Thus the total number of connected subsets of size $p$
is at most $n(3\ell_{2})^{p-1}$.
\end{proof}
\section{Zero-free region}
\label{sec:zerofree}

In this section we study the zero-free radius of the polynomial $f(\epsilon)$
defined in Eq.~\eqref{eq:f_epsilon}.
 In Section \ref{sec:lemmaproof} we prove Lemma \ref{lem:(zero-free-region)}, which establishes an $n$-independent lower bound on the zero-free radius of the polynomial $f(\epsilon)$ for constant-depth circuits. In particular, for a depth-$d$ circuit composed of two-qubit gates, the radius Eq.~\eqref{eq:epsdisc} is at least
\begin{equation}
\epsilon_0=\Omega(\gamma^{-1} 2^{-5d}),
\label{eq:depthdradius}
\end{equation}
where we used Eq.~\eqref{eq:ellconst}.  A simple example shows that this bound is tight up to constant factors in the exponential. In particular, it is easy to see that for each $d\geq 1$, the $2^d$-qubit GHZ state
\[
|\mathrm{GHZ}_{2^d}\rangle=\frac{1}{\sqrt{2}}\left(|0\rangle^{\otimes {2^d}}+|1\rangle^{\otimes {2^d}}\right)
\] 
can be prepared by a depth-$d$ circuit composed of two-qubit gates. We may choose each operator 
\begin{equation}
O_j=I+Z_j
\label{eq:oj2}
\end{equation}
 so that $\gamma =\max_j \| O_j -I\|=1$, $O_j(\epsilon)=I+\epsilon Z_j$, and
\[
f(\epsilon)=\frac{1}{2}\left((1+\epsilon)^{2^d}+(1-\epsilon)^{2^d}\right),
\]
which has zero free radius $\mathcal{O}(2^{-d})$, as can be seen by verifying that $f$ has a root at $\epsilon=(-1+e^{i\pi/2^d})(1+e^{i\pi/2^d})^{-1}$.  

While this example shows that there exist depth-$d$ circuits with zero-free radius exponentially small in $d$, we expect that such circuits are non-generic. To support this claim, in Section \ref{sec:random} we consider the zero-free radius for the polynomial $f(\epsilon)$ with operators given by  Eq.~\eqref{eq:oj2} and unitary $U$ drawn at random from any ensemble which forms a unitary $2$-design. In this case we show that with high probability the zero-free radius of $f$ is very close to $1$. 

\subsection{Proof of Lemma \ref{lem:(zero-free-region)}}
\label{sec:lemmaproof}

The proof is based on the Lov\'{a}sz Local Lemma \cite{erdHos1973problems}.

\begin{theorem}[Lov\'{a}sz Local Lemma]
Suppose $E_1,E_2,\ldots, E_m$ are events in a probability space, that each event $E_j$ is independent of all but at most $K$ of the others, and that $\mathrm{Pr}[E_j]\leq p$ for all $j$. If $p\cdot exp(1)\cdot K<1$ then 
\begin{equation}
\mathrm{Pr}\left[\cap_{j} \overline{E}_j\right]>0.
\label{eq:lll}
\end{equation}
\end{theorem}
Here $\overline{E}_j$ is the negation of event $E_j$, so Eq.~\eqref{eq:lll} is the probability that none of the events $E_1,E_2,\ldots, E_m$ occur.
We will also use the following simple fact:
\begin{lemma}
For each $j=1,\ldots, n$ and $\epsilon\in \mathbb{C}$, there is a $2$-qubit unitary $B_j(\epsilon)$ such that 
\begin{equation}
\frac{1}{\|O_j(\epsilon)\|} O_j(\epsilon)=\left(I\otimes \langle 0|\right) B_j(\epsilon) \left(I\otimes | 0\rangle\right) 
\label{eq:bj1}
\end{equation}
and
\begin{equation}
\|\left(I\otimes \langle 1|\right)B_j(\epsilon) \left(I\otimes |0\rangle\right) \|\leq 2\sqrt{\gamma |\epsilon|}.
\label{eq:bj2}
\end{equation}
\label{lem:tech}
\end{lemma}
\begin{proof}

Let $j$ and $\epsilon$ be given and define $A=\|O_j(\epsilon)\|^{-1} O_j(\epsilon)$. We may write $A=U'M$ where $U'$ is unitary and $M=(A^{\dagger}A)^{1/2}$ (polar decomposition of $A$).  Note that $\|M\|=\|A\|= 1$. Now let
\[
B=(U'\otimes I)\left(M\otimes Z+(I-M^2)^{1/2}\otimes X\right),
\]
where $Z$ and $X$ are single-qubit Pauli matrices. One can easily check that $B$ is unitary. Moreover
\[
\left(I\otimes \langle 0|\right) B \left(I\otimes | 0\rangle\right)=U'M=A,
\]
and
\begin{align}
\|\left(I\otimes \langle 1|\right)B \left(I\otimes |0\rangle\right) \|&= \|I-A^{\dagger}A\|^{1/2}\\
&\leq \left(\|I-A\|+\|A-A^{\dagger}A\|\right)^{1/2}\\
&\leq \left(2\|I-A\|\right)^{1/2},
\label{eq:adag}
\end{align}
where we used the triangle inequality along with the facts that $\|A\|=1$ and $\|I-A^{\dagger}\|=\|I-A\|$.  Now using the fact that 
\[
\|O_j(\epsilon)\|\leq 1+|\epsilon|\|O_j-I\|\leq 1+\gamma |\epsilon|,
\]
along with the triangle inequality, we get
\begin{equation}
\|I-A\|\leq \|\left(1+\gamma |\epsilon|\right)^{-1} \left(O_j(\epsilon)-I\right)\|+\|\left(1+\gamma |\epsilon|\right)^{-1}I-I\|\leq \left(\frac{2\gamma |\epsilon|}{1+\gamma |\epsilon|}\right)\leq 2\gamma|\epsilon|.
\label{eq:eyeminusa}
\end{equation}
Plugging into Eq.~\eqref{eq:adag}, we arrive at Eq.~\eqref{eq:bj2}.
\end{proof}

\begin{proof}[Proof of Lemma \ref{lem:(zero-free-region)}]
For each $j\in [n]$ consider the 2-qubit unitary $B_j(\epsilon)$ described by Lemma \ref{lem:tech}. Note that Eq.~\eqref{eq:bj1} implies (cf. Eq. ~\eqref{eq:eyeminusa})
\begin{equation}
\|\left(I\otimes \langle 0|\right) (B_j(\epsilon)-I) \left(I\otimes |0\rangle\right)\| \leq 2\gamma |\epsilon|,
\label{eq:bj3}
\end{equation}

Now for each $j\in [n]$ let us adjoin an ancilla qubit labeled $n+j$ so that $B_j(\epsilon)$ acts nontrivially on qubits $j$ and $n+j$, out of $2n$ qubits in total. Define $V_j(\epsilon)=(U^{\dagger}\otimes I)B_j(\epsilon) (U\otimes I)$ and let $S_j\subseteq [2n]$ be the qubits which it acts on nontrivially. In particular, 
\begin{equation}
S_i\subseteq \mathcal{L}(i) \cup \{n+i\} \qquad \qquad 1\leq i\leq n,
\label{eq:sj}
\end{equation}
where $\mathcal{L}(i)$ and all lightcones discussed below are with respect to the $n$-qubit unitary $U$. For future reference we also note that the unitaries $\{V_j(\epsilon)\}$ are commuting, i.e., 
\begin{equation}
[V_j(\epsilon), V_r(\epsilon)]=0 \quad 1\leq j\leq r\leq n.
\label{eq:comm}
\end{equation}
and that
\begin{equation}
\{i: j\in S_i\}\subseteq \begin{cases} \mathcal{L}_{\leftarrow}(j) \quad \mbox{if} & 1\leq j\leq n \\  j-n \quad \mbox{if}   & n+1\leq j\leq 2n.  \end{cases}.
\label{eq:sk}
\end{equation}
Indeed, from Eq.~\eqref{eq:sj} and Proposition \ref{prop:sym} we see that $j\in S_i$  only if either $j=n+i$, or the backward lightcone of $j$ contains $i$. 

Define $V(\epsilon)=\prod_{j} V_j(\epsilon)$.  Then (using Eq.~\eqref{eq:bj1})
\begin{equation}
f(\epsilon)=\left(\prod_{j=1}^{n} \|O_j(\epsilon)\|\right) \langle 0^{2n}| V(\epsilon) |0^{2n}\rangle.
\label{eq:f}
\end{equation}
Consider the probability distribution over $2n$-bit strings defined by
\[
p_{\epsilon}(z)=|\langle z|V(\epsilon)|0^{2n}\rangle|^2.
\]
To prove the lemma it suffices to show that for all $\epsilon$ satisfying $|\epsilon|\leq \epsilon_0$, where $\epsilon_0$ is given by Eq.~\eqref{eq:epsdisc},  we have
\[
p_{\epsilon}(0^{2n})>0.
\]
Indeed, using Eq.~\eqref{eq:f} we see that this implies $|f(\epsilon)|>0$, since for each $j$ we have 
\[
\|O_j(\epsilon)\|\geq 1-\gamma \epsilon_0 >0
\]
whenever $|\epsilon|\leq \epsilon_0$.

To this end, let us fix some $\epsilon$ satisfying $|\epsilon|\leq \epsilon_0$. Define events $E_j$ for $j=1,2,\ldots, 2n$ such that $E_j$ is the event that $z_j=1$ with respect to the probability distribution $p_{\epsilon}$. Then
\begin{align}
\mathrm{Pr}\left[E_j\right]&=\langle 0^{2n}| V^{\dagger}(\epsilon)|1\rangle\langle1|_j V(\epsilon)|0^{2n}\rangle\\
&=\langle 0^{2n}|Q_j|0^{2n}\rangle \qquad \text{ where } \qquad Q_j= \prod_{i: j\in S_i} V^{\dagger}_i(\epsilon)|1\rangle\langle1|_j  \prod_{i: j\in S_i} V_i(\epsilon)
\label{eq:Bk}
\end{align}
Here we used the fact that a gate $V^{\dagger}_i(\epsilon)$ such that $j\notin S_i$ has no support on qubit $j$ and thus commutes with both $|1\rangle \langle 1|_j$ (as well as all other unitaries $V^{\dagger}_r(\epsilon)$, cf. Eq.~\eqref{eq:comm}). All such gates appearing on the left can then be commuted through and cancel with their corresponding term $V_i(\epsilon)$ on the right.

The events $E_j$ and $E_k$ are independent whenever the corresponding operators $Q_j$ and $Q_k$ have disjoint support. 
Now for $1\leq j\leq n$ let $\bar{j}=j$ and for $n+1\leq j\leq 2n$ let $\bar{j}=n-j$.   The support of $Q_j$ satisfies

\begin{equation}
\mathrm{Support} (Q_j)\subseteq \mathcal{L}(\mathcal{L}_\leftarrow (\bar{j})) \cup \{n+\mathcal{L}_{\leftarrow}(\bar{j})\} \qquad \qquad 1\leq j\leq 2n,
\label{eq:supp}
\end{equation}
where we have defined $\{n+\mathcal{L}_{\leftarrow}(\bar{j})\}=\{n+r: r\in \mathcal{L}_{\leftarrow}(\bar{j})\}$.

Using Eq.~\eqref{eq:supp} we see that $Q_j$ and $Q_k$ have disjoint support unless 
\begin{equation}
\mathcal{L}(\mathcal{L}_\leftarrow (\bar{j}))\cap \mathcal{L}(\mathcal{L}_\leftarrow (\bar{k})) \neq \emptyset.
\label{eq:intersect}
\end{equation}
Using Proposition \ref{prop:sym} we see that the condition Eq.~\eqref{eq:intersect} implies
\[
\bar{k}\in \mathcal{L}_{\leftarrow}(\bar{j}, 4).
\]
Therefore each event $E_j$ is independent of all but at most $K$ others, where 
\begin{equation}
K=2\max_{1\leq i\leq n} |\mathcal{L}_{\leftarrow}(i, 4)|\leq 2\ell_4
\label{eq:dbound}
\end{equation}

We shall now upper bound the probability of each event $E_j$. First consider the case $j=n+r$ for some $1\leq r\leq n$. In this case 
\begin{equation}
\mathrm{Pr}\left[E_{n+r}\right]=\langle 0^{2n}| V_{r}^{\dagger}(\epsilon) |1\rangle\langle 1|_{n+r} V_{r}(\epsilon)|0^{2n}\rangle \leq 4\gamma |\epsilon| \qquad \qquad  1\leq r\leq n,
\label{eq:nplusl}
\end{equation}
where we used Eq.~\eqref{eq:bj2}. Next suppose $j\in\{1,2,\ldots, n\}$. In this case we have
\begin{equation}
\mathrm{Pr}\left[E_j\right]=\langle 0^{2n}|\prod_{i\in \mathcal{L}_{\leftarrow}(j)} V^{\dagger}_i(\epsilon)|1\rangle\langle1|_j  \prod_{i\in \mathcal{L}_{\leftarrow}(j)} V_i(\epsilon)|0^{2n}\rangle
=\alpha_j+\beta_j,
\label{eq:aplusb}
\end{equation}
where
\begin{equation}
\alpha_j=\langle 0^{2n}|\prod_{i\in \mathcal{L}_{\leftarrow}(j)} V^{\dagger}_i(\epsilon)\left(|1\rangle\langle1|_j\otimes |00\ldots 0\rangle\langle 00\ldots 0|_{n+\mathcal{L}_{\leftarrow}(j)}\right) \prod_{i\in \mathcal{L}_{\leftarrow}(j)} V_i(\epsilon)|0^{2n}\rangle
\label{eq:alpha}
\end{equation}
and
\begin{equation}
\beta_j=\langle 0^{2n}|\prod_{i\in \mathcal{L}_{\leftarrow}(j)} V^{\dagger}_i(\epsilon)\left(|1\rangle\langle1|_j\otimes (I-|00\ldots 0\rangle\langle 00\ldots 0|_{n+\mathcal{L}_{\leftarrow}(j)}\right)  \prod_{i\in \mathcal{L}_{\leftarrow}(j)} V_i(\epsilon)|0^{2n}\rangle.
\label{eq:beta}
\end{equation}

To upper bound $\alpha_j$, observe that each operator $I\otimes \langle0| V_i(\epsilon)I\otimes |0\rangle$ appearing in Eq.~\eqref{eq:alpha} can be approximated by the identity. In particular, Eq.~\eqref{eq:bj3} gives
\begin{equation}
\| I\otimes \langle 0|_{n+i}(V_i(\epsilon)-I) I\otimes |0\rangle_{n+i}\|\leq 2\gamma |\epsilon|.
\label{eq:vk}
\end{equation}
Eq.~\eqref{eq:vk} implies that the right-hand-side of Eq.~\eqref{eq:alpha} is close to zero (indeed, if all operators $V_i(\epsilon)$ in Eq.~\eqref{eq:alpha} were replaced by the identity then it would evaluate to zero). More precisely, we may combine Eqs.~(\ref{eq:alpha}, \ref{eq:vk}) and recursively use the triangle inequality to replace each gate $V_k(\epsilon)$ on the right by $I$. The errors add linearly, and we arrive at
\begin{align}
\alpha_j \leq 2\gamma |\epsilon|  |\mathcal{L}_{\leftarrow}(j)|.
\label{eq:abound}
\end{align}

To upper bound $\beta_j$, we expand
\[
I-|00\ldots 0\rangle\langle 00\ldots 0|_{n+\mathcal{L}_{\leftarrow}(j)}=\sum_{\stackrel{z\in \{0,1\}^{|\mathcal{L}_{\leftarrow}(j)|}}{ z\neq 00\ldots0}} |z\rangle\langle z|
\]
in Eq.~\eqref{eq:beta} and use Eqs.~(\ref{eq:bj1}, \ref{eq:bj2}) to obtain
\begin{equation}
\beta_j\leq \sum_{\stackrel{z\in \{0,1\}^{|\mathcal{L}_{\leftarrow}(j)|}}{ z\neq 00\ldots0}} (2\sqrt{\gamma |\epsilon|})^{2|z|}=(1+4\gamma|\epsilon|)^{|\mathcal{L}_{\leftarrow}(j)|}-1\leq e^{4\gamma |\epsilon| |\mathcal{L}_{\leftarrow}(j)|}-1\leq 8\gamma|\epsilon| |\mathcal{L}_{\leftarrow}(j)|.
\label{eq:bbound}
\end{equation}
where in the last line we used the facts that $e^x-1\leq 2x$ for $x\leq 1$ and $|\epsilon|\leq \epsilon_0 \leq (4\gamma |\mathcal{L}_{\leftarrow}(j)|)^{-1}$.

Putting together Eqs.(\ref{eq:abound}, \ref{eq:bbound}, \ref{eq:aplusb}, \ref{eq:nplusl}) we have the upper bound
\begin{equation}
\mathrm{Pr}[E_j]\leq 10\gamma |\epsilon| \ell_1
\label{eq:pbound}
\end{equation}
for all $1\leq j\leq 2n$.

Now from Eqs.~(\ref{eq:pbound}, \ref{eq:dbound}) we see that the Lov\'{a}sz local lemma guarantees $p_\epsilon(0^{2n})>0$ as long as 
\[
10\gamma |\epsilon| \ell_1 \cdot \exp(1)\cdot 2\ell_4  <1
\]
Upper bounding $\exp(1)\leq 3$ and rearranging, we arrive at Eq.~\eqref{eq:epsdisc}.
\end{proof}

\subsection{Zero-free region for random unitaries}
\label{sec:random}
In this section for convenience we specialize to the case $O(\epsilon)=\prod_{j=1}^{n} O_j(\epsilon)$ where
\[
O_j(\epsilon)=I+\epsilon Z_j.
\]
As in previous sections, we consider the polynomial $f(\epsilon)=\langle0^n|U^{\dagger} O(\epsilon) U|0^n\rangle$.
\begin{theorem}
Suppose $U$ is drawn from a unitary $2$-design and let $\alpha\in \{1,2,\ldots, \}$ be given. Then, with probability at least $1-n^{-\alpha}$, $f(\epsilon)$ is zero-free in a closed disk
\begin{equation}
\epsilon\leq R(n) \qquad \text{where} \qquad R(n)=1-\mathcal{O}(\log(n)/n).
\label{eq:randbound}
\end{equation}
\label{thm:rand}
\end{theorem}
Here the constant implied by the big-$\mathcal{O}$ notation depends on $\alpha$. Below we shall use the notation
$Z(s)=\prod_{j=1}^{n} Z_j^{s_j}$ and $X(s)=\prod_{j=1}^{n} X_j^{s_j}$ where $s\in \{0,1\}^n$ and $X_j, Z_j$ are Pauli operators acting on qubit $j$. We will use the following properties of unitary $2$-designs. 
\begin{lemma}
Suppose $U$ is drawn from a unitary $2$-design. Then
\[
\mathbb{E}_{U}\big[\langle 0^{2n}|U^{\dagger}\otimes U^{\dagger} \left(Z(r)\otimes Z(s)\right) U\otimes U|0^{2n}\rangle\big]=\begin{cases} 0 & , r\neq s\\ 2^{-n}\left(\frac{4^n-2^n}{4^n-1}\right) & , r=s\neq 0^n \\ 1 & , r=s=0^n \end{cases}.
\]
\label{lem:haar}
\end{lemma}
\begin{proof}
Since $U$ is drawn from a unitary $2$-design, we may WLOG evaluate the expectation value over the Haar measure.

First consider the case $r\neq s$ and assume WLOG that $r\neq 0^n$. Then we may choose an $X$-type Pauli $X(q)$ such that 
\[
\{X(q), Z(r)\}=0 \qquad \text{ and } \qquad [X(q), Z(s)]=0.
\]
Then using the invariance of the Haar measure we have
\begin{align}
\mathbb{E}_{U}\big[\langle 0^{2n}|&U^{\dagger}\otimes U^{\dagger} \left(Z(r)\otimes Z(s)\right) U\otimes U|0^{2n}\rangle\big]\\
&=\mathbb{E}_{U}\big[\langle 0^{2n}|U^{\dagger}\otimes U^{\dagger}  \left(X(q)Z(r)X(q)\otimes X(q)Z(s)X(q)\right) U\otimes U|0^{2n}\rangle\big]\\
&=- \mathbb{E}_{U}\big[\langle 0^{2n}|U^{\dagger}\otimes U^{\dagger} \left(Z(r)\otimes Z(s)\right) U\otimes U|0^{2n}\rangle\big].
\end{align}
and therefore the above quantity is zero, establishing the case $r\neq s$.

Next suppose that $r=s\neq 0^n$. In this case the quantity of interest is (letting $\mathcal{P}_n$ denote the set of $4^n$ $n$-qubit Pauli operators)
\begin{align*}
\mathbb{E}_{U}\big[\left|\langle 0^{n}|U^{\dagger}Z(s)U|0^{n}\rangle\right|^2\big]
&=\frac{1}{4^n-1}\mathbb{E}_{U}\big[\sum_{Q\in \mathcal{P}_n\setminus{I}}\langle 0^{2n}|U^{\dagger}\otimes U^{\dagger} \left(Q\otimes Q\right) U\otimes U|0^{2n}\rangle\big]\\
&=\left(\frac{4^n}{4^{n}-1}\right)\frac{1}{2^n}\mathbb{E}_{U}\big[\langle 0^{2n}|U^{\dagger}\otimes U^{\dagger} \left(SWAP\right)^{\otimes n} U\otimes U|0^{2n}\rangle\big]-\frac{1}{4^n-1}\\
&=\left(\frac{4^n}{4^{n}-1}\right)\frac{1}{2^n}-\frac{1}{4^n-1},
\end{align*}
which completes the proof of the second case. In the above we used the fact that $X\otimes X+Y\otimes Y+Z\otimes Z=2SWAP-I$ where $SWAP$ is the two-qubit unitary which permutes the qubits.  The third case $r=s=0^n$ is trivial.

\end{proof}

\begin{proof}[Proof of Theorem \ref{thm:rand}]
We may write
\[
f(\epsilon)=\sum_{k=0}^{n} c_k \epsilon^k 
\]
where
\[
c_k=\sum_{s\in \{0,1\}^n: |s|=k} \langle 0^n|U^{\dagger} Z(s) U|0^n\rangle.
\]
Note that $c_0=1$. A simple computation using Lemma \ref{lem:haar} gives
\begin{equation}
\mathbb{E}_U\left[ |c_k|^2\right]\leq \frac{1}{2^n} {n \choose k} \qquad 1\leq k\leq n.
\label{eq:eu}
\end{equation}
Now suppose $|\epsilon|\leq R\leq 1$. Then
\[
|f(\epsilon)-1|=|\sum_{k=1}^{n} c_k \epsilon^k|\leq \sum_{k=1}^{n} |c_k| R^k\leq \sum_{k=1}^{n/3} |c_k|+ \sum_{k=n/3}^{n} |c_k|R^k
\]
Applying Cauchy-Schwarz to each of the two terms on the RHS gives
\begin{equation}
|f(\epsilon)-1|\leq \sqrt{n/3} \left(\sum_{k=1}^{n/3} |c_k|^2\right)^{1/2}+\left(\sum_{k=1}^{n} |c_k|^2\right)^{1/2} \frac{R^{n/3}}{\sqrt{1-R^2}}.
\label{eq:fminus1}
\end{equation}
Using Eq.~\eqref{eq:eu} we get
\[
\mathbb{E}_U \left[\sum_{k=1}^{n/3} |c_k|^2\right]\leq \frac{1}{2^n}\sum_{k=1}^{n/3} {n \choose k}\leq 2^{(H(1/3)-1) n}\leq 2^{-0.08n},
\]
where $H(\cdot)$ is the binary entropy function. Using Markov's inequality, we have that with probability at least $1-(1/2)n^{-\alpha}$ over the choice of $U$, 
\begin{equation}
\sum_{k=1}^{n/3} |c_k|^2\leq \frac{2n^\alpha}{2^{0.08n}}.
\label{eq:cond1}
\end{equation}
Likewise we have
\[
\mathbb{E}_U \left[\sum_{k=1}^{n} |c_k|^2\right]\leq 1
\]
and with probability at least $1-(1/2)n^{-\alpha}$  over the choice of $U$, 
\begin{equation}
\sum_{k=1}^{n} |c_k|^2\leq 2n^{\alpha}
\label{eq:cond2}
\end{equation}

By a union bound we have that with probability at least $1-n^{-\alpha}$ both Eq.~\eqref{eq:cond1} and Eq.~\eqref{eq:cond2} hold. To complete the proof we show that if both of these events occur then the claimed bound on the zero-free radius of $f$ holds. Indeed, plugging Eqs.~(\ref{eq:cond1},\ref{eq:cond2}) into Eq.~\eqref{eq:fminus1} gives
\begin{equation}
|f(\epsilon)-1|\leq \left(\frac{2n^{\alpha+1}}{3}\right)^{1/2} 2^{-0.04n}+\sqrt{2n^{\alpha}} \frac{R^{n/3}}{\sqrt{1-R^2}}.
\end{equation}
Now for all sufficiently large $n$ we may choose $R(n)=1-\mathcal{O}(\log(n)/n)$ to make the RHS at most $1/2$. This establishes that $|f(\epsilon)|\geq 1/2$ for all 
$|\epsilon|\leq R(n)$ and therefore that $f$ is zero-free in this disk.
\end{proof}

\section{Additive approximation for general shallow circuits }

In this section we give a subexponential classical algorithm for estimating the absolute value $|\langle 0^n|U^{\dagger} O U|0^n\rangle|$ of the mean to a given additive error, for a tensor product observable $O=O_1\otimes O_2\ldots \otimes O_n$. In the case where each observable $O_j$ is positive semidefinite this provides a subexponential algorithm for the mean value problem.

\begin{theorem}
Let $U$ be an $n$-qubit, depth-$d$ quantum circuit and suppose that $\|O_j\|=1$ for all $j\in [n]$. There exists a classical algorithm which, given $\delta \in (0,1/2)$, computes an estimate $E \in \mathbb{R}$ such that
\[
\left|E-|\langle 0^n|U^{\dagger} O U|0^n\rangle|\right|\leq \delta.
\]
The runtime of the algorithm is upper bounded as $2^{\tilde{\mathcal{O}}\left(4^d\sqrt{n \log(\delta^{-1})}\right)}$.
\label{thm:additive}
\end{theorem}
Theorem \ref{thm:additive} is obtained as a straightforward corollary of the following algorithm for additively approximating output probabilities of constant-depth circuits.
\begin{lemma}
Let $V$ be an $n$-qubit, depth-$d$ quantum circuit. There exists a classical algorithm which, given $\delta\in (0,1/2)$, computes an estimate $q \in \mathbb{R}$ such that
\[
\left|q-|\langle 0^n|V|0^n\rangle|^2\right| \leq \delta.
\]
The runtime of the algorithm is upper bounded as $2^{\tilde{\calO}\left(2^d\sqrt{n \log(\delta^{-1})}\right)}$.
\label{lem:sub}
\end{lemma}

Let us now see how Theorem \ref{thm:additive} follows from Lemma \ref{lem:sub}. Suppose we are given $\delta$ and $U$. From Lemma \ref{lem:tech}, for each $j=1,2,\ldots, n$ we may efficiently compute a two qubit unitary $B_j$ such that 
\[
\left(I\otimes \langle 0|\right) B_j \left(I\otimes |0\rangle \right)=O_j
\]
where we used the theorem's assumption that $\|O_j\|=1$. Now consider a $2n$ qubit system where for each $j$ we adjoin an ancilla qubit $n+j$ and the unitary $B_j$ acts between these two qubits. Define $B=\bigotimes_{j=1}^{n} B_j$. We may then use the algorithm from Lemma \ref{lem:sub} with $V=(U^{\dagger}\otimes I)B (U\otimes I)$, $n'=2n$, $d'=2d+1$,  and $\delta'=0.5 \delta^2$ to obtain an estimate $q$ such that
\begin{equation}
\left|q-|\langle 0^{2n}|(U^{\dagger}\otimes I)B (U\otimes I)|0^{2n}\rangle|^2\right|\leq 0.5\delta^2.
\label{eq:qerror1}
\end{equation}
or equivalently 
\begin{equation}
\left|q-|\langle 0^{n}|U^{\dagger}O U|0^{n}\rangle|^2\right|\leq 0.5\delta^2.
\label{eq:qerror}
\end{equation}
Now if $q<0.5 \delta^2$ then the above implies that $|\langle 0^n|U^{\dagger} O U|0^n\rangle|<\delta$ and in this case we simply output $E=0$ as our $\delta$-error estimate. On the other hand if $q>0.5 \delta^2$ then we output $E=\sqrt{q}$ as our estimate; using Eq.~\eqref{eq:qerror} we get
\[
\left|\sqrt{q}-|\langle 0^n|U^{\dagger} O U|0^n\rangle|\right|\leq \frac{0.5\delta^2}{\sqrt{q}+|\langle 0^n|U^{\dagger} O U|0^n\rangle|}\leq \frac{0.5\delta^2}{\sqrt{q}}<\frac{\delta}{\sqrt{2}}.
\]

Lemma \ref{lem:sub} is a simple consequence of the following well-known fact \cite{buhrman1999bounds, de2008note}.
\begin{lemma}[\cite{buhrman1999bounds}]
Let $\delta\in (0,1/2)$ be given. There exists a univariate polynomial $g: \mathbb{R}\rightarrow \mathbb{R}$ of degree 
\begin{equation}
L=\mathcal{O}\left(\sqrt{n\log(\delta^{-1})}\right)
\label{eq:L}
\end{equation}
such that 
\begin{equation}
g(0)=1 \qquad \text{ and } \qquad |g(c)|\leq \delta \qquad \text{ for each } \quad c=1,2,\ldots, n.
\label{eq:fconditions}
\end{equation}
The coefficients of the polynomial can be computed in time polynomial in $n$. 
\label{lem:polyapprox}
\end{lemma}

We shall now review how the polynomial claimed in the Lemma is obtained in a standard way from a quantum query algorithm, see Refs.~\cite{beals2001quantum, buhrman1999bounds, de2008note} for more details. In particular, consider the problem of computing the OR of an $n$-bit string $x=x_1x_2\ldots x_n$, given quantum query access to $x$.  It is known that this function can be computed with error probability at most $\delta$ using a number of queries $T=\mathcal{O}(\sqrt{n\log(\delta^{-1})})$ \cite{buhrman1999bounds}. The probability that the quantum algorithm outputs $0$ is a multilinear polynomial $p(x_1,x_2, \ldots, x_n)$ in the input bits $x_1, x_2, \ldots, x_n$ of degree at most $2T$ \cite{beals2001quantum}. The algorithm has the feature that $p$ depends only on the Hamming weight $w=\sum_{i=1}^{n} x_i$ of the string $x$, and therefore we may write
\[
p(x_1,x_2, \ldots, x_n)=g(w)
\]
where $g:\mathbb{R}\rightarrow \mathbb{R}$ is a univariate degree $2T$ polynomial. The algorithm succeeds with probability $1$ if $x=00\ldots 0$ and errs with probability at most $\delta$ in all other cases. Therefore $g(0)=1$ and $|g(w)|\leq \delta$ for all $w=1,2,\ldots, n$. Note that in order to compute the cofficients of $g$ it suffices to evaluate it at $2T\leq n$ points $w\in \{0,1,2,\ldots, n\}$. 

The remarkable small-error $\sqrt{\log(\delta^{-1})}$ dependence \cite{buhrman1999bounds, Kahn1996inclusion}--which can also be achieved for polynomials computing symmetric functions other than OR \cite{de2008note}--is related to the fact that that we only care about the values of the polynomial $g$ at integer values of $c$. A weaker error bound scaling as $\mathcal{O}(\log(\delta^{-1}))$ can be obtained more directly using a Chebyshev polynomial \cite{linial1990approximate} (or alternatively, via a suboptimal quantum algorithm which reduces error by parallel repetition). A more direct refinement of the Chebyshev polynomial approach is used in Ref. \cite{Kahn1996inclusion} but leads to a slightly loose bound which matches Eq.~\eqref{eq:L} up to factors polylogarithmic in $n$ .

\begin{proof}[Proof of Theorem \ref{lem:sub}]

Let us define 
\[
H=\sum_{j=1}^{n} U |1\rangle\langle 1|_j U^{\dagger},
\]
where $|1\rangle\langle 1|_j$ acts nontrivially only on the $j$th qubit. Note that the eigenvalues of $H$ are the integers betwen $0$ and $n$, and that the state 
\[
|\psi\rangle=U|0^n\rangle
\]
is the unique state satisfying $H|\psi\rangle=0$. Let $\delta$ be given and consider the polynomial of degree $L=L(\delta)$ described by Lemma \ref{lem:polyapprox}. Using Eq.~\eqref{eq:fconditions} and the spectrum of $H$ we see that
\[
\| g(H)-|\psi\rangle\langle \psi| \|\leq \delta,
\]
and therefore
\[
\left| \langle 0^n | U|0^n\rangle \langle 0^n| U^{\dagger}|0^n\rangle-\langle 0^n|g(H)|0^n\rangle \right| \leq \delta.
\]
To prove the theorem it remains to show that $\langle 0^n|g(H)|0^n\rangle$ can be computed exactly using the claimed runtime.  Note that  for any positive integer $r$ we may express
\[
\langle 0^n | H^r|0^n\rangle=\langle 0^n| U \left(\sum_{j=1}^{n} |1\rangle\langle 1|_j\right)^r U^{\dagger} |0^n\rangle
\]
The right hand side is a sum of at most $n^r$ terms of the form 
\begin{equation}
\langle 0^n| U |11\ldots 1\rangle\langle 11\ldots 1|_{S} U^{\dagger} |0^n\rangle
\label{eq:terms}
\end{equation}
where $S\subset [n]$ satisfies $|S|\leq r$. Since $U$ has depth $d$, the operator
\[
U |11\ldots 1\rangle\langle 11\ldots 1|_{S} U^{\dagger}
\]
acts nontrivially on at most $2^d|S|$ qubits. Therefore each term Eq.~\eqref{eq:terms} can be computed exactly using a runtime $2^{\mathcal{O}(2^d r)}$, and  $\langle 0^n | H^r|0^n\rangle$ can be computed with runtime $n^r 2^{\mathcal{O}(2^d r)}$. Since $g(H)$ is a polynomial of degree $L$ with efficiently computable coefficients, we may compute $\langle 0^n |g(H)|0^n\rangle$ using runtime
\[
\mathrm{poly}(n)+Ln^L 2^{\mathcal{O}(2^dL)}=2^{\tilde{\mathcal{O}}(2^d\sqrt{n\log(\delta^{-1})})},
\]
where the first term on the LHS is the time used to compute the coefficients of the polynomial, and the second term is the time used to compute $\langle 0^n | H^r|0^n\rangle$ for $1\leq r\leq L$.
\end{proof}

\section{Additive approximation  for 2D and 3D circuits}
\label{sec:2D}

In this section we consider tensor product observables 
$O=O_1\otimes O_2\otimes \cdots \otimes O_n$, 
where $O_j$ are arbitrary hermitian single-qubit operators satisfying
\be
\label{TPO2}
 \|O_j\|\le 1.
\ee
As before,  our goal is to estimate  the mean value
$\mu = \la 0^n |U^\dag O U|0^n\ra$.
We prove the following.
\begin{theorem}
\label{thm:MPS}
Consider a 2D grid of $n$ qubits. Suppose $U$ is a depth-$d$ quantum circuit composed
of nearest-neighbor two-qubit gates. 
There exists a probabilistic classical algorithm that 
computes an approximation $\tilde{\mu}$ satisfying
$|\tilde{\mu} - \mu|\le \delta$ with probability at least $2/3$.
The algorithm has runtime scaling as $n\delta^{-2}2^{\calO(d^2)}$.
\end{theorem}
\begin{proof}
Define operators
\be
\label{Qj}
Q_j=U^\dag (O_j\otimes I_{\mathsf{else}}) U.
\ee
Here $I_{\mathsf{else}}$ applies the identity operator to all qubits in $[n]\setminus \{j\}$.
Obviously, the operators $Q_1,\ldots,Q_n$ pairwise commute. Furthermore, $Q_j$ acts nontrivially only within a lightcone of radius $d$ centered at the
$j$-th qubit. We show an example of such lightcone for $d=2$ in Fig.~\ref{fig:MPS1}. It will be convenient to coarse-grain the lattice into super-sites with
local Hilbert space of dimension
 $D=2^{4d^2}$. Each super-site represents a block of qubits of size $2d\times 2d$.
An example for $d=2$ is shown in Fig.~\ref{fig:MPS1}.
\begin{figure}[hb!]
\centerline{\includegraphics[height=4cm]{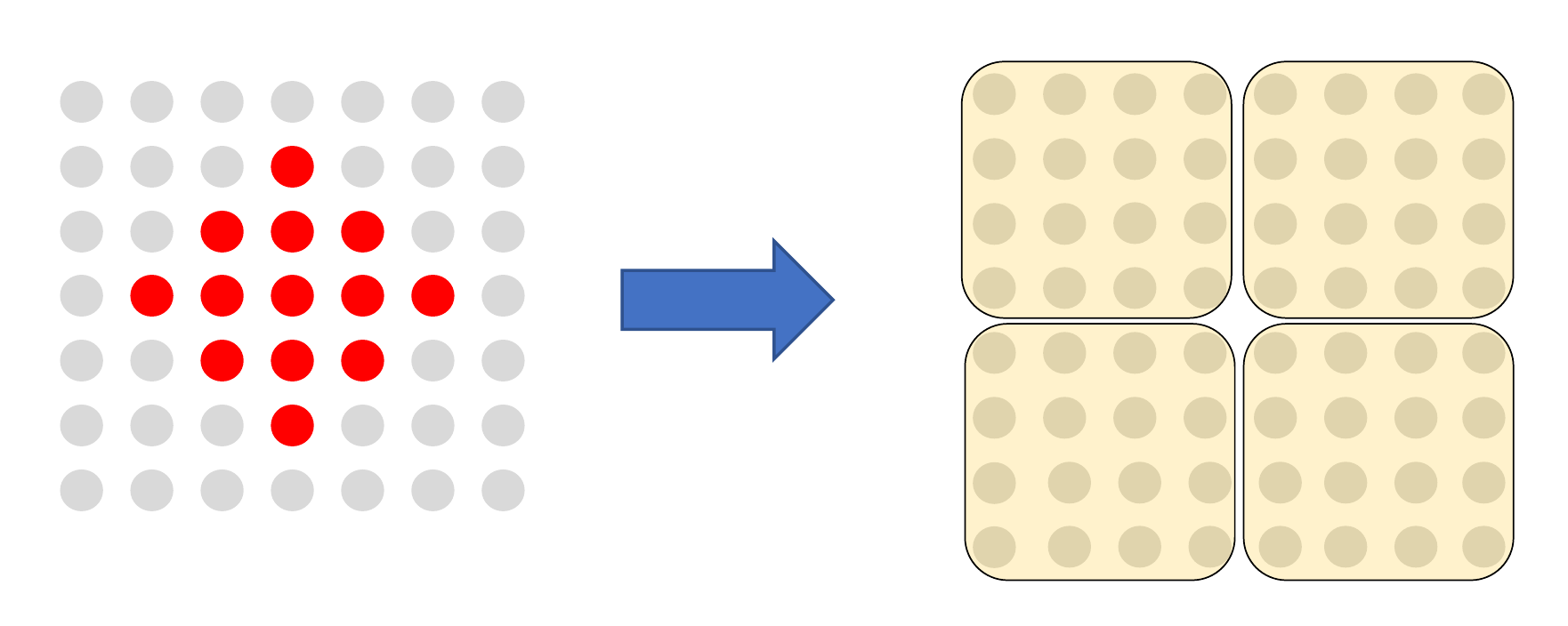}}
\caption{{\em Left}: qubits live at sites of the 2D square lattice.
The lightcone of a single qubit  generated by a depth-$2$ circuit  is highlighted in red. 
{\em Right}: the coarse-grained lattice $\Lambda$.
Each $4\times 4$ block of sites becomes a super-site of the coarse-grained
lattice. A plaquette is a $2\times 2$ cell spanning four adjacent  super-sites.
The support of any operator $Q_j$ is covered by a single plaquette.
}
\label{fig:MPS1}
\end{figure}

Let $\Lambda$ be the coarse-grained lattice. It has linear size $L\times L$, where
\[
L\approx \frac{\sqrt{n}}{2d}.
\]
We shall label sites $u\in \Lambda$ by pairs of integers $(i,j)$, where $1\le i,j\le L$. Let us agree that $i$ and $j$ label rows and columns of $\Lambda$
respectively.
Define a plaquette $p(i,j)$ as a $2\times 2$ cell of $\Lambda$
spanning super-sites $(i,j)$, $(i+1,j)$, $(i,j+1)$, and $(i+1,j+1)$.
Let $Q_{i,j}$ be the product of all operators $Q_s$ whose support is fully contained
in the plaquette $p(i,j)$. If the support of  $Q_s$ is contained
in more than one plaquette, assign $Q_s$ to one of them to avoid duplication. 
Then
\be
Q_1Q_2\cdots Q_n = \prod_{1\le i,j\le L-1} Q_{i,j}.
\ee
Here we noted that all $Q$'s pairwise commute, so the order does not matter.
This yields 
\be
\mu = \la \Psi_0|\Psi_1\ra \quad \mbox{where} \quad 
|\Psi_b\ra = \prod_{(i,j)\, : \,  j = {b \pmod 2} }\; \;
 Q_{i,j} |0^n\ra.
\ee
The product ranges over $1\le i,j\le L-1$ to ensure that all plaquettes
$p(i,j)$ are fully inside the lattice. We claim that
for each $b\in \{0,1\}$ there exists a linear order on the set of 
$n$ qubits such that 
 the state
$\Psi_b$ is Matrix Product State (MPS) with a small bond dimension
that depends only on the circuit depth $d$.
Below we prove the claim for the state $\Psi_1$ (exactly the same arguments
apply to $\Psi_0$).

Let $C_1,C_2,\ldots,C_L$ be the consecutive columns of $\Lambda$.
Assume for simplicity that $L=2K$ is an even integer.
A direct inspection shows that none of the plaquettes $p(i,j)$ 
with odd coordinate $j$ crosses the boundary between vertical strips
\[
A_\alpha=C_{2\alpha-1} C_{2\alpha}, \quad \alpha=1,2,\ldots,K.
\]
The strips $A_\alpha$ are shown in Fig.~\ref{fig:MPS2} for $K=3$.
For example, all plaquettes $p(i,1)$ are fully contained in the strip $A_1$,
plaquettes $p(i,3)$ are fully contained in $A_2$ etc.
Thus $\Psi_1$ is a tensor product of $K$ single-strip states
associated with $A_1,\ldots,A_K$,
\[
|\Psi_1\ra = |\Psi_1(A_1)\ra\otimes |\Psi_1(A_2)\ra \otimes \cdots \otimes |\Psi_1(A_K)\ra,
\]
where
\[
|\Psi_1(A_\alpha)\ra = \prod_{(i,j)\, : \, p(i,j)\subseteq A_\alpha}\; \; Q_{i,j}|0\ra_{A_\alpha}.
\]
We claim that each single-strip state $\Psi_1(A_\alpha)$ is an MPS
with bond dimension $\chi\le D^3$. 
Recall that $D$ is the local Hilbert space dimension of  each super-site. 
\begin{figure}[hb]
\centerline{\includegraphics[height=5cm]{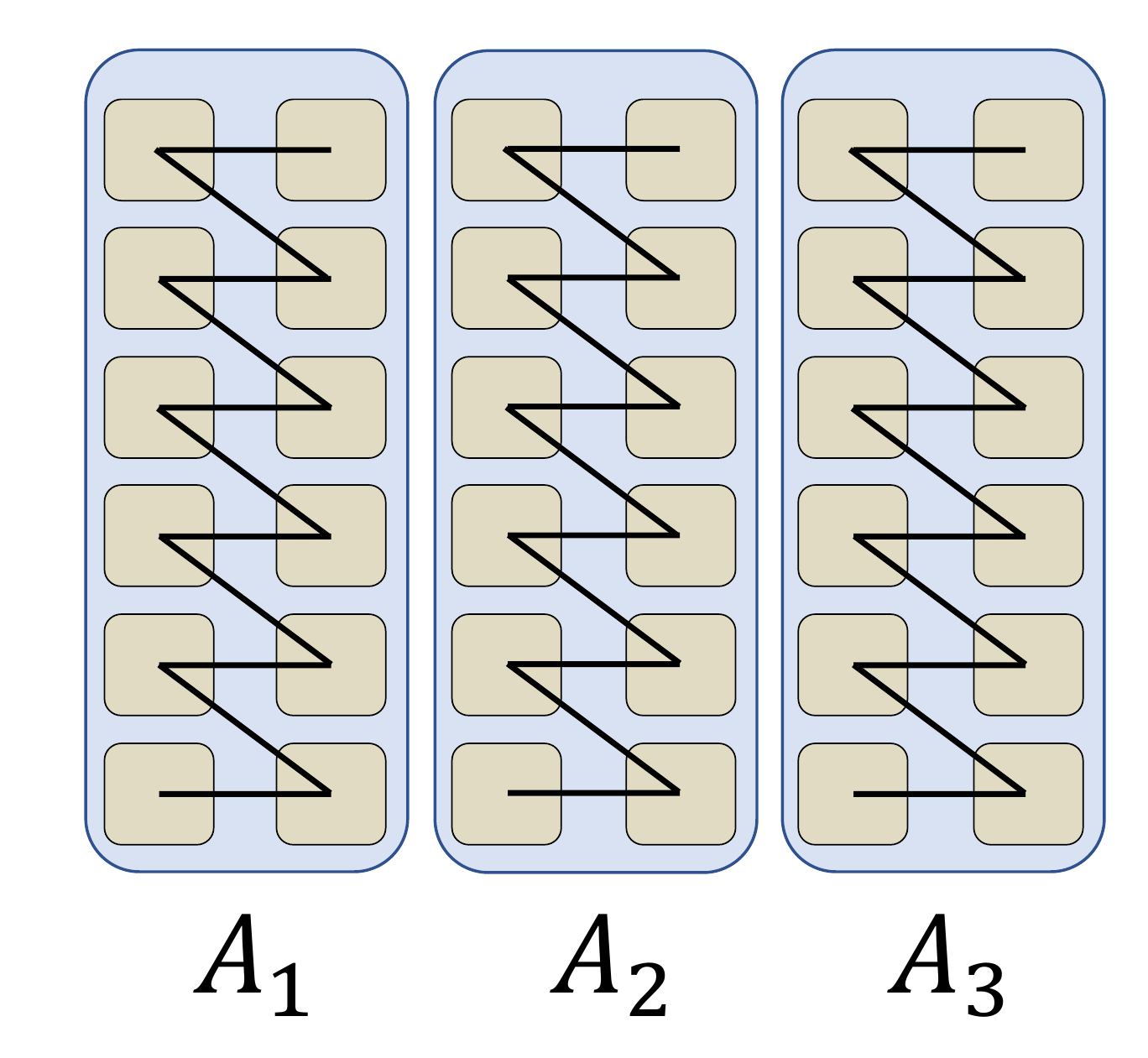}}
\caption{Coarse-grained lattice of size $L=5$
and the snake-like linear order that define Matrix Product States
$\Psi_1(A_1)$, $\Psi_1(A_2)$, $\Psi_1(A_3)$.
}
\label{fig:MPS2}
\end{figure}

Indeed, consider a fixed strip $A_\alpha$ and 
choose a snake-like linear order such that
$A_\alpha= \{1,2,\ldots,2L\}$, see Fig.~\ref{fig:MPS2}.
Consider any bipartite cut $A_\alpha =A_\alpha'A_\alpha''$,
where $A_\alpha'$ and $A_\alpha''$ are consecutive blocks of super-sites.
A direct inspection shows that there are at most two plaquettes
$p(i,j)$ that are contained in the strip $A_\alpha$ and cross the
chosen cut, see Fig.~\ref{fig:MPS3}.
\begin{figure}[h]
\centerline{\includegraphics[height=5cm]{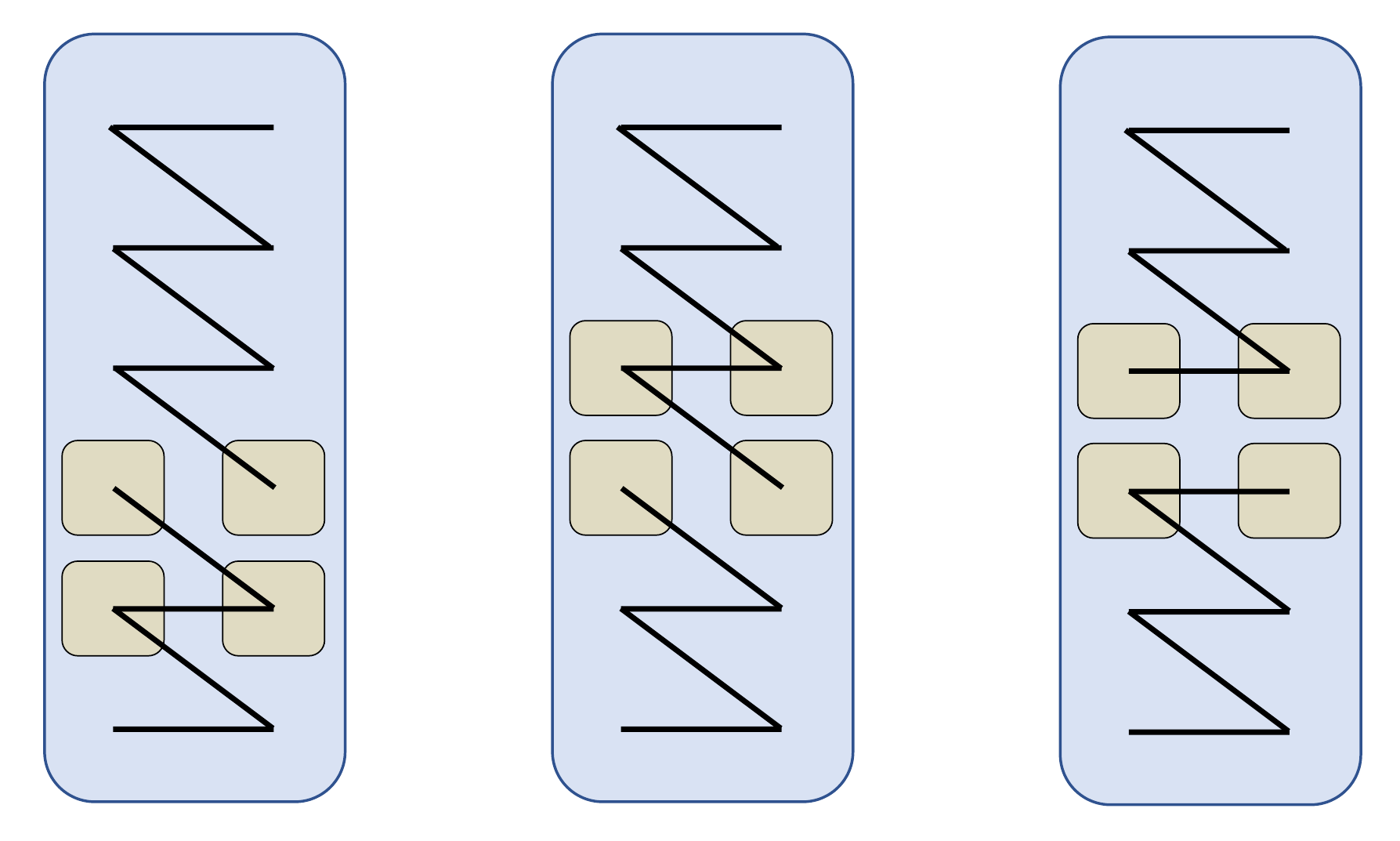}}
\caption{Examples of a bipartite cut of the chain $A_\alpha$
and  plaquettes crossing the cut. For each cut there are at most
two plaquette operators $Q_{i,j}$ capable of creating entanglement
across the cut. The state $\Psi_1(A_\alpha)$ is
obtained from the all-zero basis state by applying all plaquette operators
$Q_{i,j}$ contained in $A_\alpha$. The above shows that
$\Psi_1(A_\alpha)$ is a Matrix Product State with a small bond dimension,
$\chi\le D^3$.
}
\label{fig:MPS3}
\end{figure}
Furthermore, the corresponding plaquette
operators $Q_{i,j}$ act nontrivially on at most three super-sites located next to the cut,
see Fig.~\ref{fig:MPS3}.
Thus the Schmidt rank of $\Psi_1(A_\alpha)$ across the chosen cut is at most
$\chi=D^3$. Accordingly, $\Psi_1(A_\alpha)$ is an MPS with bound dimension $\chi$.
The same applies to the full state $\Psi_1$
since the latter is a tensor product of the states $\Psi_1(A_\alpha)$.
By the symmetry, the same arguments apply to the state $\Psi_0$.
It should be emphasized that the linear orders  in the MPS representation
of $\Psi_0$ and $\Psi_1$ are not the same. 
Thus, the desired mean value can be written as
\be
\mu = \la \Psi_0 |W|\Psi_1\ra,
\ee
where $W$ is a permutation of $n$ qubits that 
accounts for the difference between  linear orders used by $\Psi_0$ and $\Psi_1$.
The MPS description of the states $\Psi_b$ can be computed
starting from the circuit $U$ and the list of observables $O_j$
using the well-known algorithms~\cite{vidal2003efficient,yoran2006classical,jozsa2006simulation}.
It takes time $n\cdot \mathrm{poly}(\chi) = n2^{\calO(d^2)}$.

Let $\gamma_{b}=\| \Psi_{b}\|$, where $b=0,1$. 
Note that $\gamma_{b}\le 1$ since we assumed 
$\|O_j\|\le 1$ for all $j$ and $U$ is a unitary operator.
Furthermore, 
one can compute $\gamma_{b}$ in time
$\calO(n\chi^3)$ using the standard MPS algorithms~\cite{MPSreview}. Define normalized states
$|\Phi_b\ra = \gamma_{b}^{-1} |\Psi_{b}\ra$. Then 
\be
\mu = \gamma_0 \gamma_1 \la \Phi_0|W|\Phi_1\ra.
\ee
Define a probability distribution 
\[
\pi(x) =|\la x|\Phi_0\ra|^2, \qquad x\in [D]^{L^2}
\]
and a function 
\[
F(x)=\gamma_0 \gamma_1 \frac{\la x |W|\Phi_1\ra}{\la x|\Phi_0\ra}
\]
which is well-defined whenever $\pi(x)>0$.
Then $\mu$ coincides with the mean value of $F(x)$ over the distribution $\pi(x)$, 
\be
\mu = \sum_x \pi(x) F(x).
\ee
The random variable $F(x)$ has the variance
\be
\mathsf{Var}(F)\le \sum_x \pi(x) |F(x)|^2 =
(\gamma_0\gamma_1)^2 \sum_x |\la x|\Phi_1\ra|^2 =(\gamma_0\gamma_1)^2 \le 1.
\ee
Define an empirical mean value 
$\tilde{\mu}=S^{-1} \sum_{i=1}^S F(x^i)$,
where $x^1,\ldots,x^S\in [D]^{L^2}$ are independent samples from the distribution $\pi(x)$
and the number of samples is $S=3\delta^{-2}$.
By Chebyshev inequality, $|\tilde{\mu}-\mu|\le \delta$ with probability at least $2/3$. 
It remains to notice that any amplitude
$\la x|\Phi_b\ra$ can be computed using the standard MPS
algorithms~\cite{MPSreview} in time $\calO(n\chi^3)$.
Accordingly, one can compute $F(x)$ for any given string $x$ in time $\calO(n\chi^3)$. 
The probability distribution $\pi(x)=|\la x|\Phi_0\ra|^2$ can be sampled
in time $n\cdot \mathrm{poly}(\chi)$ using the algorithm of Ref.~\cite{jozsa2006simulation},
see Theorem~1 thereof. To summarize, the overall cost
of approximating $\mu$ is $n\delta^{-2} \mathrm{poly}(\chi) = n\delta^{-2} 2^{\calO(d^2)}$.
\end{proof}

Suppose now that $U$ is a geometrically local depth-$d$ quantum circuit
on a three-dimensional grid of $n$ qubits
of linear size $n^{1/3}$.
 Define the coarse-grained lattice $\Lambda$
as a two-dimensional grid, see Fig.~\ref{fig:MPS1},
where each super-site represents a block of qubits of size $2d\times 2d\times n^{1/3}$.
The lattice $\Lambda$ has size $L\times L$ with $L\approx n^{1/3}/2d$.
Now each super-site has the local Hilbert space
of dimension $D=2^{4d^2n^{1/3}}$.
Repeating exactly the same arguments as above one 
gets a representation $\mu=\la \Psi_0|W|\Psi_1\ra$, where $\Psi_b$ are
MPSs with bond dimension $\chi\le D^3 = 2^{\calO(d^2 n^{1/3})}$
and $W$ is a permutation of $n$ qubits.
Thus one can estimate $\mu$ within an additive error $\delta$
in time $n\delta^{-2} \mathrm{poly}(\chi)=\delta^{-2}2^{\calO(d^2 n^{1/3})}$.

\section{Acknowledgments}
The authors thank Robert Koenig and Kristan Temme for helpful discussions.
DG acknowledges the support of the Natural Sciences and Engineering Research Council of Canada (NSERC) under Discovery grant number RGPIN-2019-04198. DG is a CIFAR fellow in the Quantum Information Science program.
SB and RM acknowledge the support of the IBM Research
Frontiers Institute and funding from the MIT-IBM Watson AI Lab 
under the project Machine Learning in Hilbert space.

\bibliographystyle{unsrt}
\bibliography{mybib}

\end{document}